\documentclass[envcountsame]{llncs}
    \usepackage{amsmath}
    \usepackage{amsxtra}
    \usepackage{amstext}
    \usepackage{amssymb}
    \usepackage{latexsym}
    \usepackage{dsfont} 
\newcommand{\GF}[1]{{\mathbb F}_{#1}}

\newcommand{\Tr}[2][1]{\mathbf{Tr}_{#1}^{#2}}

\begin{document}
\title{On the Boomerang Spectrum of Power Permutation
$X^{2^{3n}+2^{2n}+2^{n}-1}$ over $\GF{2^{4n}}$ and Extraction of  Optimal Uniformity Boomerang Functions}
\author{Kwang Ho Kim\inst{1,2} \and Sihem Mesnager\inst{3} \and Ye Bong Kim\inst{4}}

\institute{Institute of Mathematics, State Academy of Sciences,
Pyongyang, Democratic People's Republic of Korea\\
\and PGItech Corp., Pyongyang, Democratic People's Republic of Korea\\ \and Department of Mathematics, University of Paris VIII, F-93526 Saint-Denis, University Sorbonne Paris Cit\'e, LAGA, UMR 7539, CNRS, 93430 Villetaneuse and Telecom Paris, Polytechnic Institute of Paris, 91120 Palaiseau, France.\\
\email{smesnager@univ-paris8.fr}\\
\and O Jung Hup Chongjin University of Education, Chongjin,
Democratic People's Republic of Korea} \maketitle

\begin{abstract} 
A substitution box (S-box) in a symmetric primitive is a mapping $F$ that takes $k$ binary inputs and whose image is a binary $m$-tuple for some positive integers $k$ and $m$, which is usually the only nonlinear element of the most modern block ciphers. Therefore, employing S-boxes with good cryptographic properties to resist various attacks is significant. For power permutation $F$ over finite field $\GF{2^k}$, the multiset of
 values $\beta_F(1,b)=\#\{x\in \GF{2^k}\mid
F^{-1}(F(x)+b)+F^{-1}(F(x+1)+b)=1\}$  for $b\in \GF{2^k}$ is called
the boomerang spectrum of $F$. The maximum value in the boomerang
spectrum is called boomerang uniformity. This paper determines the
boomerang spectrum of the power permutation
$X^{2^{3n}+2^{2n}+2^{n}-1}$ over $\GF{2^{4n}}$. The boomerang
uniformity of that power permutation is $3(2^{2n}-2^n)$. However, on
a large subset  $\{b\in \GF{2^{4n}}\mid \mathbf{Tr}_n^{4n}(b)\neq 0\}$ of $\GF{2^{4n}}$
of cardinality $2^{4n}-2^{3n}$ (where $ \mathbf{Tr}_n^{4n}$ is the (relative) trace function from $\GF{2^{4n}}$ to $\GF{2^{n}}$), we prove that  the studied function $F$  achieves the optimal
boomerang uniformity $2$.
 It is known that obtaining such functions is a challenging problem.
 More importantly, the set
of $b$'s giving this value is explicitly determined for any value in the boomerang spectrum.

\end{abstract}


\section{Introduction}
Let $n$ be a positive integer and $\mathbb F_{2^n}$ be the finite field of order $2^n$. We denote by $\mathbb F_{2^n}^*$ the multiplicative cyclic group of non-zero elements of $\mathbb F_{2^n}$.  Vectorial (or multi-output) Boolean functions (that is $(n,m)$-functions from $\mathbb{F}_{2^n}$ to $\mathbb{F}_{2^m}$, where $n$ and $m$ are two positive integers) are widely applied to block ciphers' design in cryptography. When they are used in symmetric cryptography, vectorial Boolean functions correspond to substitution boxes (for short, \emph{S-boxes}). They are fundamental parts of block ciphers and an essential component of symmetric key algorithms that perform the substitution. Block ciphers are cornerstones of our cryptographic landscape today and are used to ensure security for a significant fraction of our daily communication. Their design and analysis are well advanced, and with today's knowledge designing a secure block cipher is a problem that is largely considered solved. The reader can consult the recent book  \cite{Carlet-book-2021}. Particularly, $(n,n)$-vectorial functions, which are bijective, have an extreme interest, especially in symmetric cryptography. Typically, block ciphers use a permutation as an S-box during the encryption process, while the compositional inverse of the S-box is used during the decryption process. Many block ciphers use permutations $F$ defined over the finite field $\mathbb F_{2^n}$ to itself. When permutations $F$ are used as S-boxes inside a block cipher, the differential uniformity $\delta_F$ of a permutation $F$ (used as an S-box inside a cryptosystem) measures the resistance of the block cipher against the differential cryptanalysis. The differential uniformity (see \cite{Nyberg1994}) of a vectorial Boolean function $F: \mathbb{F}_{2^n}\rightarrow \mathbb{F}_{2^n}$ is defined by
$
\delta_F=\max_{a,b\in \mathbb{F}_{2^n}, a\neq 0}\verb"DDT"_F(a,b),
$
where $\verb"DDT"_F(a,b)$ is the entry at $(a,b)\in \left(\mathbb{F}_{2^n}\right)^2$ of the difference distribution table
$
\verb"DDT"_F(a,b)=\#\{x\in \mathbb{F}_{2^n},\ F(x+a)+F(x)=b\}.
$

When $F$ is used as an S-box inside a cryptosystem, the smaller the value $\delta_F$ is, the better $F$ to the resistance against differential attack. Typically, the optimal functions satisfy $\delta_F=2$ and are called almost perfect nonlinear (APN).  Another important cryptanalytical technique on block ciphers is the boomerang attack, introduced by Wagner~\cite{W99} in 1999, a differential cryptanalysis variant. Sometimes,  the boomerang attack is suitable to attack a cipher when no significant differential probability exists for the whole cipher. The boomerang uniformity $\beta_{F}$, (a notion introduced by Boura and Canteaut in \cite{BC18}),
measures its contribution to
the resistance to the boomerang attack, which is a powerful
cryptanalysis technique introduced by Wagner in \cite{W99} against
block ciphers involving S-boxes and considered as an extension of
the differential attack \cite{BS91}. The smaller $\beta_{F}$ is the
better is the resistance of $F$ to this attack.  The optimal value is $2$. It seems much harder to construct permutation functions with $2$-uniform boomerang uniformity. Over binary field $\GF{2^n}$, the best resistance belongs to functions with
$\beta_{F}=2$, these APN functions.

For odd values of $n$, there are known families of APN permutations.
While, for $n$ even, no APN permutation exists for $n = 4$ and, up
to some equivalence, there exists only one example of APN
permutation over $\GF{2^6}$ (\cite{Dillon10}), and with respect to
the affine equivalence (for which the boomerang uniformity is
invariant), these known APN permutations can be divided into four affine
equivalence classes (\cite{Carlet-book-2021}). The existence of more APN permutations
remains a big open problem. One of the modern and important applications of designing permutation polynomials over binary finite fields in the domain of vectorial Boolean functions in the context of block ciphers in symmetric cryptography is the constructions of permutations with boomerang uniformity $4$. The high importance of such families of permutations for real applications comes from the difficulty of finding APN permutations in even dimension (known as the Big APN Problem), making researchers contented with those permutations having boomerang uniformity $4$. APN permutations offer maximal resistance to differential and boomerang attacks, but there are extremely difficult to construct despite many efforts and recent advances.

 As a particular class of functions over finite fields of
characteristic $2$, power functions, namely, monomial functions, have
been extensively studied in the last decades due to their simple
algebraic form and lower implementation cost in hardware
environment.  Their particular algebraic structure makes determining their differential properties easier to handle.

Consequently, it  is  pretty interesting to study the boomerang spectrums
of none-APN power permutations. However, as of now, there exist only very few families of none-APN
power permutations over $\GF{2^n}$ with known boomerang spectrum,
which are summarized in Table~\ref{table1}.
\begin{table}
\begin{center}
\caption{None-APN power permutations $F(X)=X^d$ over $\GF{2^n}$ with
known boomerang spectrum and their boomerang
uniformities}\label{table1}
\begin{tabular}{|c|c|c|c|c|}\hline
No. &$d$                   &Condition                 &$\beta_{F}$             &Refs\\
\hline
1   &$2^{n}-2$             &any $n$                   &$2\leq \beta_{F}\leq6$  &\cite{BC18,JLLQ22,EM22}\\
\hline
2   &$2^{k}+1$             &$s=\gcd(k,n)$             &$2^s$ or $2^s(2^s-1)$   &\cite{BC18,EM22,HPS22}\\
\hline
3   &$2^{k+1}-1$           &$n=2k, k>1$               &$2^k+2$                 &\cite{YZL22}\\
\hline
4   &$2^{3k}+2^{2k}+2^k-1$ &$n=4k$                    &$3(2^{2k}-2^k)$         &This paper\\
\hline
\end{tabular}
\end{center}
\end{table}

Very recently, Budaghyan, Calderini, Carlet, Davidova and,  Kaleyski (\cite{BCCDK22}) have investigated on
the differential spectrum of power functions $F(X)=X^d$ with
$d=\sum_{i=1}^{k-1}2^{in}-1$ over $\GF{2^{kn}}$. This class of power
functions includes some famous functions as special cases. When
$n=1$ and $d=\sum_{i=1}^{k-1}2^{i}-1=(2^k-1)-2$, $F(X)=X^{d}$ is
equivalent to the well-known inverse function over $\GF{2^k}$, which
has been widely used in practical cryptosystems in symmetric cryptography. If $k=3$, then
$F(X)=X^{2^{2n}+2^n-1}$ is equivalent to $X^{2^{2n}-2^n+1}$ over
$\GF{2^{3n}}$, which is a Kasami function \cite{Kasami71}. If $k=5$,
then $F(X)=X^{2^{4n}+2^{3n}+2^{2n}+2^n-1}$ over $\GF{2^{5n}}$ is
precisely the Dobbertin power permutation \cite{Dobbertin01}. For
$k=2$, the differential property of $F(X)=X^d$ has been studied in
\cite{BCC11,Budaghyan05}. For $k=4$, the differential spectrum of
$F(X)=X^d$ has been determined in
\cite{KM22,KM22-1,LWZT20,LWZT22,TLWZTJ22}.

This paper aims to describe explicitly the boomerang spectrum of
the power permutation $X^{2^{3n}+2^{2n}+2^{n}-1}$ over
$\GF{2^{4n}}$. It is organized as follows. In Section \ref{main-result}, we
present the boomerang spectrum of the power permutation
$X^{2^{3n}+2^{2n}+2^n-1}$ over $\GF{2^n}$ in
Theorem~\ref{maintheorem}. In Subsection \ref{sec:decomp-boom-unif},
we state a decomposition formula on the boomerang uniformity of a
function $F$ on, which strongly relies upon our proof. In
Subsection~\ref{sec:toolkit}, we give basic results that we shall need
in this paper. Finally, in Section~\ref{sec:proof-main-result}, we
give the proof of Theorem~\ref{maintheorem}. Computer experiments have also checked our main results for small values of $n$.

\section{Main Theorem}\label{main-result}
This ar
Let $n$ be a positive integer.  We begin by giving some notation that we
shall use in this paper.
\begin{itemize}
\item $\#\Omega$ is the cardinality of a given finite set $\Omega$.
\item Let $Q$ be a power of $2$,
  $\mathbb F_Q$ denotes a finite field
  of cardinality $Q$.
\item $\GF{Q}^*:=\GF{Q}\setminus \{0\}$.
\item $q=2^n$.  
\item Let $k$ and $l$ be two integers such that $l$ is a divisor of
  $k$. Define two functions $\mathbf{Nr}_l^k$ and $\mathbf{Tr}_l^k$
  from $\GF{2^k}$ to $\GF{2^l}$ by
  $$\mathbf{Nr}_l^k(x):=\prod_{i=0}^{\frac kl-1}x^{2^{li}},\quad \mathbf{Tr}_l^k(x):=\sum_{i=0}^{\frac kl-1}x^{2^{li}}.$$
  Each of these functions maps $\GF{2^k}$ onto
$\GF{2^l}$, then called norm and trace, respectively.

\item $\mathfrak{F}_{2^m,i}:=\{\mathbf{Tr}_1^m(x)=i, x\in \GF{2^m}\}$.
\item $\mu_{m}:=\{x\in \GF{q^4}\mid x^{m}=1\}$
\item $\mu_m^\star:=\mu_m\setminus \{1\}$.
\item $\mathfrak{S}_2:=\left\{b\in \GF{q^4}\mid
    \mathbf{Tr}_n^{4n}(b)\neq 0,  \mathbf{Tr}_1^{n}\left(
\frac{\mathbf{N}_n^{4n}(b^{q+1}+1)}{\left[(b+b^{q^2})\mathbf{Tr}_n^{4n}(b)\right]^2}\right)
=\frac{(b^{q+1}+1)^{q^2+1}}{(b+b^{q^2})^{q+1}}\right\}$.
\item $\mathfrak{S}_2':=\left\{b\in \GF{q^4}\mid
     \mathbf{Tr}_n^{4n}(b)\neq 0, \mathbf{Tr}_1^{n}\left(
\frac{\mathbf{N}_n^{4n}(b^{q+1}+1)}{\left[(b+b^{q^2})\mathbf{Tr}_n^{4n}(b)\right]^2}\right)
=\frac{(b^{q+1}+1)^{q^2+1}}{(b+b^{q^2})^{q+1}}+1\right\}$.
\end{itemize}
In this paper, we completely determines the boomerang spectrum  (or Boomerang Connectivity Table) $\beta_F(1,b)$ 
of the
power permutation $F(X)=X^{q^3+q^2+q-1}$ over $\GF{q^4}$.

\begin{theorem}\label{maintheorem} For the power permutation $F(X)=X^{q^3+q^2+q-1}$ over
$\GF{q^4}$, it holds
\[\beta_F(1,b)=
\begin{cases} 3q(q-1), \text{ if $b\in \mu_{q+1}\setminus\{1\}$,}\\
2q^2, \text{ if $b\in \GF{q}\setminus\GF{2}$,}\\
2q^2-3q, \text{ if $b\in \GF{q^2}\setminus\{\GF{q}\cup\mu_{q+1}$\},}\\
q^2, \text{ if $b=1$,}\\
q^2-2q, \text{ if $b\in \GF{q^4}\setminus \{\GF{q^2}\cup \mathfrak{S}_2 \cup \mathfrak{S}_2'\}$,}\\
2, \text{ if $b\in \mathfrak{S}_2$,}\\
0, \text{ if $b\in \mathfrak{S}_2'$.}\\
\end{cases}\]
\end{theorem}

\section{Preliminaries}\label{sec-prel}

\subsection{Decomposition of the boomerang uniformity}\label{sec:decomp-boom-unif}

 Cid,  Huang, Peyrin, Sasaki, and  Song (\cite{Cid18}) introduced the concept of Boomerang
Connectivity Table (BCT) for a permutation $F$ over $\GF{2^n}$. Afterwards,   in \cite{BC18}, Boura and Canteaut introduced the notion of boomerang uniformity.

\begin{definition} Let $F$ be a permutation over  $\GF{Q}$ and
$a,b$ in $\GF{Q}$.  Define
$$\beta_F(a,b):=\#\{x\in \GF{Q}\mid F^{-1}(F(x)+b)+F^{-1}(F(x+a)+b)=a\}.$$
 The multiset
$$\mathfrak{B}=\{\beta_F(a,b)\}_{a,b\in \GF{Q}^*}$$ is called the Boomerang Connectivity Table or the boomerang spectrum of $F$.
 The maximum value in the  boomerang spectrum
$$\max_{a,b\in \GF{Q}^*}\beta_F(a,b)$$ is called the boomerang
uniformity of $F$ over $\GF{Q}$.\qed
\end{definition}

As noticed by Li, Qu, Sun,  and Li. presented in \cite{LQSL19}, $\beta_F(a,b)$ is equal to
the number of solutions $(x,y)\in \GF{Q}^2$ to the system
\begin{equation*}
\begin{cases}
F(X)+F(Y)=b\\ F(X+a)+F(Y+a)=b.
\end{cases}
\end{equation*}
Consequently, using this equivalent definition for the boomerang spectrum, it is possible to consider functions that are not
permutations. Further, note that this system can also be rewritten as
\begin{equation}\label{iniBU_eq}
\begin{cases}
F(X)+F(Y)=b\\ F(X)+F(X+a)=F(Y)+F(Y+a),
\end{cases}
\end{equation}
which we will consider as the basic system of equations for computing
$\beta_F(a,b)$ in this paper.

In the particular case where $F$ is a power function, it is well known
that $ \beta_F(a,b) = {\beta}_F\big(1,b(F(a))^{-1}\big) $ for any $a$
and $b$ in $\GF{Q}^*$. Therefore,
$$\mathfrak{B}=\{\beta_F(1,b)\}_{b\in \GF{Q^*}},$$
The computation of  $\beta_F(1,b)$ can be divided into sub-computations, introducing
$\tilde{\beta}_F(1,b,c)$ as the number of the solutions $(x,y)$ in
$\GF{q^4}^2$ of the equation system
\begin{equation}\label{BU_eq}
\begin{cases}
F(X)+F(Y)=b\\ F(X)+F(X+1)=F(Y)+F(Y+1)=c.
\end{cases}
\end{equation}
Then, it holds 
\begin{proposition}\label{prop-decomposition}
  For any $b$ in $\GF{Q}^*$, it holds
\[
\beta_F(1,b)=\sum_{c\in \GF{q^4}} \tilde{\beta}_F(1,b,c).
\]

\end{proposition}

\subsection{A toolkit}\label{sec:toolkit}
First of all, since
$\gcd(q^4-1,q^3+q^2+q-1)=\gcd((q-1)(q+1)(q^2+1),q^3+q^2+q-1)=1$, the
power function $X^{q^3+q^2+q-1}$ is a permutation over $\GF{q^4}$.

The following theorem is the main result of \cite{KM22,KM22-1}, that
we shall use as the starting point of this paper.
\begin{theorem}[{\cite[Theorem 5]{KM22-1}}]\label{BaseThm}
The number of solutions to the equation
$$X^{q^3+q^2+q-1}+(X+1)^{q^3+q^2+q-1}=b$$ in $\GF{q^4}$
  equals to
  \begin{enumerate}
  \item\label{item:1} $q^2$ if $b=1$,
  \item\label{item:2} $q^2-q$ if $b\in\mu_{q+1}^\star$,
  \item\label{item:3} $2$ if $b\in\mathfrak{S}_2$,
  \item\label{item:3} $0$ if $b\in\GF{q^4}\setminus\{\mathfrak{S}_2\cup \mu_{q+1}^\star\cup \{1\}\}$.
  \end{enumerate}\qed
 \end{theorem}

We shall very frequently use the following facts throughout this
paper.

\begin{proposition}[{\cite[Proposition 1]{KM19}}]\label{Aprop}
  \label{prop:decomposition}
  Let $m$ be a positive integer. Then, every element $z$ of
  $\GF{2^m}^*$ is written exactly twice as
  $z=c+\frac 1c$ where $c\in\GF{2^m}^\star$
  if $\mathbf{Tr}_1^m(\frac{1}{z})=0$ and
  $c\in \mu_{2^m+1}^{\star}$ if
  $\mathbf{Tr}_1^m(\frac{1}{z})=1$.
\end{proposition}

\begin{lemma}[{\cite[Lemma 4]{KM22}}]
 $\GF{q^4}^*$ can be composed  as $\GF{q^4}^*=\mu_{q-1}\cdot\mu_{q+1}\cdot \mu_{q^2+1}$.
\end{lemma}

\begin{lemma}[{\cite[Lemma 5]{KM22}}]\label{Fcase}
It holds true
\[
 \{x\in \GF{q^4}\mid
x^{q^3+q^2+q-1}+(x+1)^{q^3+q^2+q-1}=1\}=\GF{q^2}.
\]
\end{lemma}

\begin{lemma}\label{Scard} Let $m\geq 2$, $Q=2^m$, $\gamma\in \GF{Q}\setminus\GF{2}$ and $S_{Q,\gamma,i,j}:=\{Y\in \GF{Q}\mid
\mathbf{Tr}_1^{m}\left(Y\right)=i \text{
  and } \mathbf{Tr}_1^{m}\left(\gamma Y\right)=j\}$ for
$i,j\in \{0,1\}$. Then \[\#S_{Q,\gamma,i,j}=\frac{Q}{4}\] for any
$i,j\in \{0,1\}$ and $\gamma\in \GF{Q}\setminus\GF{2}$.
\end{lemma}
\begin{proof}
 Observe that $S_{Q,\gamma,0,0}\cup S_{Q,\gamma,0,1}=\mathfrak{F}_{Q,0}$
and $S_{Q,\gamma,0,0}\neq \mathfrak{F}_{Q,0}$ because every $Y\in
S_{Q,\gamma,0,0}$ satisfies equation
$\gamma^{2^{m-1}}\mathbf{Tr}_1^{m}(Y)+\mathbf{Tr}_1^{m}(\gamma Y)=0$
which has degree $2^{m-2}=\frac{\#\mathfrak{F}_{Q,0}}{2}$ in terms
of $Y$ as $\gamma\in \GF{Q}\setminus\GF{2}$. So, $S_{Q,0,1}\neq \emptyset$
and $S_{Q,\gamma,0,1}=Y_0+S_{Q,\gamma,0,0}$ for any element $Y_0\in
S_{Q,\gamma,0,1}$. Therefore
$\#S_{Q,\gamma,0,0}=\#S_{Q,\gamma,0,1}$. On the other hand,
$\#S_{Q,\gamma,0,0}+\#S_{Q,\gamma,0,1}=\#\mathfrak{F}_{Q,0}=\frac{Q}{2}$.
Thus, $\#S_{Q,\gamma,0,0}=\#S_{Q,\gamma,0,1}=\frac{Q}{4}$ and by
symmetry $\#S_{Q,\gamma,1,0}=\#S_{Q,\gamma,0,1}=\frac{Q}{4}$. Hence,
$\#S_{Q,\gamma,1,1}=Q-\#S_{Q,\gamma,0,0}-\#S_{Q,\gamma,0,1}-\#S_{Q,\gamma,1,0}=Q-\frac{3Q}{4}=\frac{Q}{4}.$
\qed
\end{proof}

Through similar discussions, one can also derive the following statement.

\begin{lemma}\label{Scard_ext} Let $m\geq 2$, $Q=2^m$, $\gamma_1, \gamma_2, \gamma_1+\gamma_2\in \GF{Q}\setminus\GF{2}$ and $S_{Q,\gamma_1,\gamma_2, i,j, k}:=\{Y\in \GF{Q}\mid
\mathbf{Tr}_1^{m}\left(Y\right)=i \text{
  and } \mathbf{Tr}_1^{m}\left(\gamma_1 Y\right)=j \text{ and } \mathbf{Tr}_1^{m}\left(\gamma_2 Y\right)=k\}$ for
$i,j, k\in \{0,1\}$. Then \[\#S_{Q,\gamma_1,\gamma_2, i,j,
k}=\frac{Q}{8}\] for any $i,j,k\in \{0,1\}$.\qed
\end{lemma}

In order to state our main result, we first need  the following description of the union set $\mathfrak{S}_2 \cup \mathfrak{S}_2'$.
\begin{proposition}
It holds
$$\mathfrak{S}_2 \cup \mathfrak{S}_2'=\{b\in \GF{q^4} \mid \mathbf{Tr}_{n}^{4n}(b)\neq0\}.$$
\end{proposition}
\begin{proof}
Set
$A=\frac{\mathbf{N}_n^{4n}(b^{q+1}+1)}{\left[(b+b^{q^2})\mathbf{Tr}_n^{4n}(b)\right]^2}$
and $U=\frac{(b^{q+1}+1)^{q^2+1}}{(b+b^{q^2})^{q+1}}+1.$ One can check straightforwardly that  $A+A^q=U+U^2,$ i.e. $\mathbf{Tr}_1^{n}(A+A^2)=U+U^2.$
Hence, whenever $\mathbf{Tr}_{n}^{4n}(b)\neq0$, either
$\mathbf{Tr}_1^{n}(A)=U$ or $\mathbf{Tr}_1^{n}(A)=U+1$ holds
true.\qed
\end{proof}

\begin{corollary}
$$\#\mathfrak{S}_2=\#\mathfrak{S}_2'=\frac{q^4-q^3}{2}.$$
\end{corollary}
\begin{proof}
Corollary 6 of \cite{KM22-1} shows
$\#\mathfrak{S}_2=\frac{q^4-q^3}{2}.$ Since $\#\{b\in \GF{q^4} \mid
\mathbf{Tr}_{n}^{4n}(b)\neq0\}=q^4-q^3$, it follows
$\#\mathfrak{S}_2'=\frac{q^4-q^3}{2}.$\qed
\end{proof}

\section{Proof of the main result}\label{sec:proof-main-result}

According to Proposition~\ref{prop-decomposition} and
Theorem~\ref{BaseThm}, we have
$$\beta_F(1,b)=\tilde{\beta}_F(1,b,1)+\sum_{c\in
\mathfrak{S}_2}\tilde{\beta}_F(1,b,c)+\sum_{c\in
\mu_{q+1}^\star}\tilde{\beta}_F(1,b,c).$$

The two first terms are easy to compute (we omit the calculation of
the second term and give an outline of the proof of the calculation of the first term).

\begin{proposition}\label{sum2}
  Let $b\in \GF{q^4}$. Then,
  $$\sum_{c\in \mathfrak{S}_2}\tilde{\beta}_F(1,b,c)=\begin{cases}
    2\cdot\#\mathfrak{S}_2=q^4-q^3,
    \text{ if } b=0,\\
    \tilde{\beta}_F(1,b,b)=2,
    \text{ if } b\in \mathfrak{S}_2,\\
    0, \text{ otherwise.}
\end{cases}$$
\end{proposition}
\begin{proposition}\label{12b}    $\tilde{\beta}_F(1,b,1)=\begin{cases}q^2,
\text{ if } b\in \GF{q^2},\\
0, \text{ otherwise.} \end{cases}$
\end{proposition}
\begin{proof}
  By Proposition~\ref{Fcase}, the set of all $\GF{q^4}$-solutions to
  $F(X)+F(X+1)=1$ is $\GF{q^2}$, and $F(x)+F(y)=b$ for
  $(x,y)\in \GF{q^2}^2$ can be rewritten as
  $x+y=b^{\frac{1}{2q}}$. Therefore,  the first item follows. \qed
\end{proof}

The main difficulty is, therefore, to determine the third term
$\sum_{c\in \mu_{q+1}^\star}\tilde{\beta}_F(1,b,c)$.


To this end, we make a critical remark that we shall use in our computations.

\begin{remark}
Set
$l:=(q-1)(q^2+1)=q^3-q^2+q-1.$ For $c\in \mu_{q+1}^\star$, from the
proof of Lemma 7 in \cite{KM22}, one can see that every
$\GF{q^4}-$solution to $F(X)+F(X+1)=c$ is expressed as
$x=\frac{1}{1+zt}$ with $(z,t)\in \mu_{q-1}\times\mu_{q^2+1}^\star$
and it holds $x^{l}=c$. Therefore, for solutions of \eqref{BU_eq}, we
set
\begin{equation}\label{x0}
x:=\frac{1}{1+z_1t_1}, y=\frac{1}{1+z_2t_2},
\end{equation}
where $(z_i,t_i)\in \mu_{q-1}\times\mu_{q^2+1}^\star$ for $i=1,2$
and $x^l=y^l=c$. Then, since $F(x)+F(y)=x^l(x+y)^{2q^2}$, we can
compute $\sum_{c\in \mu_{q+1}^\star}\tilde{\beta}_F(1,b,c)$ as the solution
number of
\begin{equation}\label{aux_eq0}
\begin{cases}b=\frac{1}{(1+z_1t_1)^l}\left(\frac{1}{1+z_1t_1}+\frac{1}{1+z_2t_2}\right)^{2q^2}\\
\frac{1}{(1+z_1t_1)^l}=\frac{1}{(1+z_2t_2)^l},
\end{cases}
\end{equation}
which is a system of equations in four variables $z_1,z_2, t_1, t_2$
where $(z_i,t_i)\in \mu_{q-1}\times\mu_{q^2+1}^\star$ for $i=1,2$.
\end{remark}

Based on this remark, we now state all the possible values of
$\sum_{c\in \mu_{q+1}^\star}\tilde{\beta}_F(1,b,c)$ regarding to the possible
values of $b$. Since the computation is long and depends on the values
of $b$, we present our results in separate subsections.

\subsection{ Case $b=0$}

This is the case when $z_1t_1=z_2t_2$, that is, $z_1=z_2$ and
$t_1=t_2$. Therefore,
\begin{lemma}
If $b=0$, then $\sum_{c\in
\mu_{q+1}^\star}\tilde{\beta}_F(1,b,c)=(q-1)q^2$.\qed
\end{lemma}

\subsection{Transformation of the problem when $b\in \GF{q^2}^*$}
Now, set $M_i:=z_i+\frac{1}{z_i}$ and $T_i=t_i+\frac{1}{t_i}$. Then
$M_i\in \GF{q}$, $T_i\in \GF{q^2}$. Further, $M_i=0$ (resp. $T_i=0$)
if and only if $z_i=1$ (resp. $t_i=1$). By Proposition~\ref{Aprop},
when $M_i\neq 0$ and $T_i\neq 0$, it is always true
$$\Tr{n}\left(\frac{1}{M_i}\right)=0 \text{   and   } \Tr{2n}\left(\frac{1}{T_i}\right)=1.$$

Note that $b\in \GF{q^2}$ if and only if
$\frac{1}{1+z_1t_1}+\frac{1}{1+z_2t_2}\in \GF{q^2}$, because $x^l\in
\mu_{q+1}\subset \GF{q^2}$ for any $x\in \GF{q^4}$. The following
proposition restates the condition $b\in \GF{q^2}$ in terms of $M_i$
and $T_i$.
\begin{proposition}\label{q^2cond}
$\frac{1}{1+z_1t_1}+\frac{1}{1+z_2t_2}\in \GF{q^2}$ if and only
if $M_1T_2=M_2T_1$.
\end{proposition}
\begin{proof} $\frac{1}{1+z_1t_1}+\frac{1}{1+z_2t_2}\in \GF{q^2}$ if and only
if
$\frac{1}{1+z_1t_1}+\frac{1}{1+z_2t_2}=\left(\frac{1}{1+z_1t_1}+\frac{1}{1+z_2t_2}\right)^{q^2}$,
i.e.,
$\frac{1}{1+z_1t_1}+\frac{1}{1+z_2t_2}=\frac{1}{1+z_1t_1^{-1}}+\frac{1}{1+z_2t_2^{-1}}$,
or equivalently,
$\frac{1}{1+z_1t_1}+\frac{1}{1+z_1t_1^{-1}}=\frac{1}{1+z_2t_2}+\frac{1}{1+z_2t_2^{-1}}$,
which can be rewritten
$\frac{z_1t_1+z_1t_1^{-1}}{1+z_1^2+z_1t_1+z_1t_1^{-1}}=\frac{z_2t_2+z_2t_2^{-1}}{1+z_2^2+z_2t_2+z_2t_2^{-1}}$,
i.e. $\frac{T_1}{M_1+T_1}=\frac{T_2}{M_2+T_2}$ which is equivalent
to $M_1T_2=M_2T_1$.\qed
\end{proof}
\begin{proposition} Let $(z_i,t_i)\in \mu_{q-1}\times\mu_{q^2+1}^\star$ for $i=1,2$.
\begin{enumerate}
\item If $M_1T_2=M_2T_1$ and $M_1M_2\neq 0$, then
$\frac{1}{(1+z_1t_1)^l}=\frac{1}{(1+z_2t_2)^l}$.
\item
When $M_1M_2= 0$, it holds
$\frac{1}{(1+z_1t_1)^l}=\frac{1}{(1+z_2t_2)^l}$ if and only if
$T_1^{q-1}=T_2^{q-1},$ i.e., $T_2=eT_1$ for some $e\in \mu_{q-1}.$
\end{enumerate}
\end{proposition}
\begin{proof} For $i=1,2$,
\begin{equation}\label{ztd}
\frac{1}{(1+z_it_i)^l}=\frac{(1+z_it_i)^{q^2+1}}{(1+z_it_i)^{q(q^2+1)}}=
\frac{(1+z_it_i)(1+z_it_i^{-1})}{\left((1+z_it_i)(1+z_it_i^{-1})\right)^{q}}=\frac{M_i+T_i}{M_i+T_i^q}.
\end{equation}
If $M_1T_2=M_2T_1$ and $M_1M_2\neq 0$, then it follows
$M_1T_2^q=M_2T_1^q$ and subsequently $T_1^{q-1}=T_2^{q-1}$, i.e.,
$T_1T_2^q=T_2T_1^q$. Thus, it holds
$\frac{M_1+T_1}{M_1+T_1^q}=\frac{M_2+T_2}{M_2+T_2^q}$.

If $M_1=M_2= 0$, then
$\frac{M_1+T_1}{M_1+T_1^q}=\frac{M_2+T_2}{M_2+T_2^q}$ holds if and
only if $T_1^{q-1}=T_2^{q-1}$, i.e. $T_2=eT_1$ for some $e\in
\mu_{q-1}.$\qed
\end{proof}

If $M_1T_2=M_2T_1$, then
\begin{align*}
&\left(\frac{1}{1+z_1t_1}+\frac{1}{1+z_2t_2}\right)^{2q^2}=\left(\frac{1}{1+z_1t_1}+\frac{1}{1+z_2t_2}\right)\cdot\left(\frac{1}{1+z_1t_1}+\frac{1}{1+z_2t_2}\right)^{q^2}\\
&=\left(\frac{1}{1+z_1t_1}+\frac{1}{1+z_2t_2}\right)\cdot\left(\frac{1}{1+z_1t_1^{-1}}+\frac{1}{1+z_2t_2^{-1}}\right)\\
&=\frac{(z_1t_1+z_2t_2)(z_1t_1^{-1}+z_2t_2^{-1})}{(1+z_1t_1)(1+z_1t_1^{-1})(1+z_2t_2)(1+z_2t_2^{-1})}\\
&=\frac{z_1^2+z_2^2+z_1z_2(t_1t_2^{-1}+t_1^{-1}t_2)}{(1+z_1^2+z_1(t_1+t_1^{-1}))(1+z_2^2+z_2(t_2+t_2^{-1}))}=\frac{M_{12}+T_{12}}{(M_1+T_1)(M_2+T_2)},
\end{align*}
where $$M_{12}=z_1z_2^{-1}+z_1^{-1}z_2 \text{ and }
T_{12}=t_1t_2^{-1}+t_1^{-1}t_2.$$

Thus, when $b\in \GF{q^2}$, Equation~\eqref{aux_eq0} is equivalent
to
\begin{equation}\label{aux_eq_qq}
b=\frac{M_{12}+T_{12}}{(M_1+T_1)^q(M_2+T_2)} \text{ and }
\begin{cases} M_1M_2\neq 0, M_1T_2=M_2T_1, z_1t_1\neq z_2t_2\text{, or, } \\
z_1=z_2=1, T_1^{q-1}=T_2^{q-1}, t_1\neq t_2.
\end{cases}
\end{equation}

Note that the equality~\eqref{ztd} and the second condition of
\eqref{aux_eq_qq} ensure that
$$(M_1+T_1)^q(M_2+T_2)=(M_1+T_1)(M_2+T_2)^q\in \GF{q}.$$

\subsection{Solving the problem when $b\in \GF{q}^*$}
If $b=\frac{M_{12}+T_{12}}{(M_1+T_1)^q(M_2+T_2)}\in \GF{q}$, then
$T_{12}\in \GF{q}$, that is, $T_{12}=0$ and $$t_1=t_2.$$
Subsequently $T_1=T_2$, and from $M_1T_2=M_2T_1$ it follows
$M_1=M_2$, i.e., $z_1=z_2$ or $z_1=z_2^{-1}$. As $b\neq0$ and
$T_{12}=0$, then $M_{12}\neq 0$, i.e., $z_1\neq z_2$, and so it
holds $$z_1=z_2^{-1}\neq 1$$ and $M_{12}=M_1^2$. Hence, in this
case, Equation~\eqref{aux_eq_qq} transforms into
\begin{equation}\label{aux_eq_q}
b=\frac{M_{1}^2}{(M_1+T_1)^{q+1}}=\frac{1}{(1+\frac{T_1}{M_1})^{q+1}}.
\end{equation}
 If $b=1$, then $T_1^{q+1}+T_1^qM_1+T_1M_1=0$ and so
$\Tr{2n}\left(\frac{1}{T_1}\right)=\Tr{n}\left(\frac{1}{T_1}+\frac{1}{T_1^q}\right)=\Tr{n}\left(\frac{1}{M_1}\right)=0$
which leads to a contradiction. Therefore, with Proposition~\ref{12b}, we
get
\begin{lemma}
If $b=1$, then $\sum_{c\in \mu_{q+1}^\star}\tilde{\beta}_F(1,b,c)=0$ and
${\beta}_F(1,b)=q^2.$\qed
\end{lemma}

When $b\neq 1$, we need the following result.
\begin{lemma}\label{bbbb} The mapping
$(z,t)\in \GF{q}^\star\times \mu_{q^2+1}^\star\longmapsto
\left(1+\frac{t+t^{-1}}{z+z^{-1}}\right)^{q+1}\in \GF{q}^\star$ is
$q^2-$to$-1$ from $\GF{q}^\star\times \mu_{q^2+1}^\star$ onto
$\GF{q}^\star$.
\end{lemma}
\begin{proof} Fix any  $e\in \GF{q}^\star$. Then, we shall
compute the number $\mathcal{N}$ of $(z,\lambda,t)\in
\GF{q}^\star\times \mu_{q+1}\times \mu_{q^2+1}^\star$ such that
$$1+\frac{T}{M}=e\lambda.$$ Since
$\Tr{2n}\left(\frac{1}{T}\right)=\Tr{n}\left(\frac{1}{T}+\frac{1}{T^q}\right)
=\Tr{n}\left(\frac{1}{(e\lambda+1)M}+\frac{1}{(e\lambda^{-1}+1)M}\right)
=\Tr{n}\left(\frac{L}{(E+L)M}\right), $ where
$L=\lambda+\lambda^{-1}$ and $E=e+e^{-1}$, thanks to
Proposition~\ref{Aprop} we have
$$\mathcal{N}=8\times \#\{(M,L)\in \GF{q}^*\times\GF{q}^*\mid \Tr{n}\left(\frac{1}{M}\right)=0, \Tr{n}\left(\frac{1}{L}\right)=1, \Tr{n}\left(\frac{L}{(E+L)M}\right)=1\}.$$

 By introducing  new  denotations as follows: $\alpha=\frac{1}{E}$,
$X=\frac{1}{L}+\frac{1}{E}$ and $Y= \frac{1}{M}$, the above
expression transforms into
\[\mathcal{N}=8\times\#\{(X, Y)\in \GF{q}^*\times\GF{q}^* \mid \Tr{n}(X)=1, \Tr{n}(Y)=0,
\Tr{n}\left(\frac{\alpha Y}{X}\right)=1\}.\]

Fix any $X\in \mathfrak{F}_{q,1}$ and set $\gamma:=\frac{\alpha}{X}$.
Since $\alpha\in \mathfrak{F}_{q,0}$, we have $\gamma\notin \GF{2}$.
Then Lemma~\ref{Scard} shows $\#\{Y\in \GF{q}\mid \Tr{n}(Y)=0,
\Tr{n}\left(\gamma Y\right)=1\}=\frac{q}{4}$ and so we have
\[\mathcal{N}=8\times \#\mathfrak{F}_{q,1}\times \frac{q}{4}=8\times \frac{q}{2}\times \frac{q}{4}=q^2.\] This completes the proof.\qed
\end{proof}
As an immediate consequence of Lemma~\ref{bbbb} and
Proposition~\ref{12b}, we get the following result.
\begin{lemma}
If $b\in \GF{q}^\star$, then $\sum_{c\in
\mu_{q+1}^\star}\tilde{\beta}_F(1,b,c)=q^2$ and ${\beta}_F(1,b)=2q^2.$
\end{lemma}

\subsection{Solving the problem when $b\in \mu_{q+1}^\star$}

In this case, we begin from deriving a restatement of
Equation~\eqref{aux_eq_qq}.
 By raising
the both sides of $b=\frac{M_{12}+T_{12}}{(M_1+T_1)^q(M_2+T_2)} $ to
the $(q-1)-$th powering , we get
$$b^{q-1}=(M_{12}+T_{12})^{q-1}\left(\frac{M_1+T_1}{M_2+T_2}\right)^{q-1}=(M_{12}+T_{12})^{q-1}.$$
Since
\begin{align*}
&\left(\frac{M_{12}+T_{12}}{(M_1+T_1)^q(M_2+T_2)}\right)^{q+1}=\left(\frac{M_{12}+T_{12}}{(M_1+T_1)(M_2+T_2)}\right)^{q+1}\\
&=\left(\frac{M_{12}+T_{12}}{M_1M_2+T_1T_2}\right)^{q+1}=\left(\frac{M_{12}+T_{12}}{M_{12}+T_{12}+M_{11}+T_{11}}\right)^{q+1},
\end{align*}
where $$M_{11}=z_1z_2+(z_1z_2)^{-1} \text{ and }
T_{11}=t_1t_2+(t_1t_2)^{-1},$$ it follows that
$\frac{M_{12}+T_{12}}{(M_1+T_1)^q(M_2+T_2)}\in \mu_{q+1}$ is
 equivalent to $\frac{M_{12}+T_{12}}{M_{12}+T_{12}+M_{11}+T_{11}}\in \mu_{q+1}.$
Therefore, Equation~\eqref{aux_eq_qq} transforms into
\begin{equation}\label{aux_eq_mu_q_not1}
b^{q-1}=(M_{12}+T_{12})^{q-1} \text{ and }
\left(1+\frac{M_{11}+T_{11}}{M_{12}+T_{12}}\right)^{q+1}=1\text{ and
}
\begin{cases} M_1M_2\neq 0, M_1T_2=M_2T_1\text{, or, } \\
z_1=z_2=1, T_1^{q-1}=T_2^{q-1}.
\end{cases}
\end{equation}
Note that in this case, $t_1\neq t_2$ is  insured because otherwise, 
$T_{12}=0$ and $b^{q-1}=1$. In the following lemma, we compute the
solution number of \eqref{aux_eq_mu_q_not1}.
\begin{lemma}
If $b\in\mu_{q+1}^\star$ , then $\sum_{c\in
\mu_{q+1}^\star}\tilde{\beta}_F(1,b,c)=2q^2-3q$ and $\beta_F(1,b)=3q^2-3q.$
\end{lemma}
\begin{proof}
First, when $M_1=M_2=0$, i.e., $z_1=z_2=1$, we should compute
\[\mathcal{N}_1=\#\left\{(t_1,t_2)\in \mu_{q^2+1}^\star \times \mu_{q^2+1}^\star\mid
b^{q-1}=T_{12}^{q-1},
\left(\frac{T_{11}+T_{12}}{T_{12}}\right)^{q+1}=1,\left(\frac{T_1}{T_2}\right)^{q-1}=1
\right\},\] that is, the number of $(t_1,t_2)\in \mu_{q^2+1}^\star
\times \mu_{q^2+1}^\star$ such that
\begin{equation}\label{0000}
T_{12}=be_1^{-1}, T_2=e_2T_1 \text{ for some } e_1, e_2 \in \GF{q}^*
\text{ and } \left(\frac{T_{11}+T_{12}}{T_{12}}\right)^{q+1}=1 .
\end{equation} The condition
 $T_{12}=be_1^{-1}$ for some $e_1\in \GF{q}^*$ requires $1=\Tr{2n}\left(\frac{e_1}{b}\right)=\Tr{n}\left(\frac{e_1}{b}+e_1b\right)$, i.e.,
\begin{equation}\label{0001}
\Tr{n}\left(\frac{e_1}{B}\right)=1,
\end{equation} where $B=\frac{1}{b+b^{-1}}.$ Now,  as a
solution of $T_{12}=be_1^{-1}$, we will fix
$$\alpha=\frac{t_2}{t_1}\in \mu_{q^2+1}^\star.$$ Then, the condition
$T_2=e_2T_1$ gives $\alpha
t_1+\alpha^{-1}t_1^{-1}=e_2t_1+e_2t_1^{-1},$ i.e.,
\begin{equation}\label{0002}
t_1=\sqrt{\frac{\alpha^{-1}+e_2}{\alpha+e_2}}
\end{equation}
 which lies in
$\mu_{q^2+1}^\star$ for every $e_2\in \GF{q}^*$.

If $e_2=1$, then $t_2=\frac{1}{t_1}$, $T_{11}=0$ and so
$\left(\frac{T_{11}+T_{12}}{T_{12}}\right)^{q+1}=1$ is satisfied. In
this case, the number of possible $(t_1, t_1^{-1})$'s is
$\mathcal{N}_{11}=2\times \frac{q}{2}=q$.

Now, assume $e_2\neq 1$. Using $T_{12}=\alpha+\alpha^{-1}$ and
$T_{11}=\alpha t_1^2+(\alpha t_1^2)^{-1}=\frac{1+\alpha
e_2}{\alpha+e_2}+\frac{\alpha+e_2}{1+\alpha e_2}=\frac{(1+\alpha^2)(
1+e_2^2)}{(\alpha+e_2)(1+\alpha e_2)}=\frac{1}{A+E_2}$, where
$A=\frac{1}{T_{12}}$ and $E_2=\frac{1}{e_2+e_2^{-1}}$, one can
easily verify that
$\left(\frac{T_{11}+T_{12}}{T_{12}}\right)^{q+1}=1$ is equivalent to
$\left(\frac{E_2}{A+E_2}\right)^{q+1}=1$, i.e.,
$E_2=\frac{A^{q+1}}{A+A^q}=\frac{1}{T_{12}+T_{12}^q}=e_1B$.
Therefore, in this case the number of $(t_1,t_2)\in
\mu_{q^2+1}^\star \times \mu_{q^2+1}^\star$ satisfying \eqref{0000}
equals four times the number of $e_1\in \GF{q}^*$ such that
$\Tr{n}(\frac{e_1}{B})=1$ and $\Tr{n}(e_1 B)=0$, which is
$\mathcal{N}_{12}=4\times\#S_{q,B^2,1,0}=q$. Thus,
$$\mathcal{N}_1=\mathcal{N}_{11}+\mathcal{N}_{12}=2q.$$

Next, we should be able to compute the number $\mathcal{N}_2$ of
$(z_1,t_1),(z_2,t_2)\in \left(\GF{q}^\star\times
\mu_{q^2+1}^\star\right)^2$ such that
\begin{equation}\label{0003}
b^{q-1}=(M_{12}+T_{12})^{q-1},\left(1+\frac{M_{11}+T_{11}}{M_{12}+T_{12}}\right)^{q+1}=1,
T_2=\frac{M_2}{M_1}T_1 .
\end{equation}

If $M_{12}=0$, i.e. $z_1=z_2$, then from $T_2=\frac{M_2}{M_1}T_1$ it
follows $t_2=t_1^{-1}$ and $T_{11}=0$. We can set $T_{12}:=be_1^{-1}$
for some $e_1\in \GF{q}^*$ due to the condition
$b^{q-1}=(M_{12}+T_{12})^{q-1}$. Then, the condition
$\left(1+\frac{M_{11}+T_{11}}{M_{12}+T_{12}}\right)^{q+1}=1$ gives
$\left(\frac{T_{12}+M_{11}}{T_{12}}\right)^{q+1}=1$, i.e.,
$\frac{1}{M_{11}}=\frac{1}{T_{12}^q+T_{12}}=Be_1$. To ensure that there exists corresponding $z_1$'s, we have the condition
$\Tr{n}(Be_1)=0$. The condition
$\Tr{2n}\left(\frac{1}{T_{12}}\right)=1$ gives
$\Tr{n}(e_1\frac{1}{B})=1.$ Therefore,  the number of possible $e_1$'s is
$\#S_{q,B^2,1,0}=\frac{q}{4}$ and there are  $q$ corresponding
$(z_1, t_1)$'s. Thus, the number of solutions of \eqref{0003} with
$M_{12}=0$ is $$\mathcal{N}_{21}=q.$$

Now, assume $M_{12}\neq 0$, i.e. $z_1\neq z_2$. The first condition
of \eqref{0003} can be written as $1+\frac{T_{12}}{M_{12}}=be_1^{-1}$
for some $e_1\in \GF{q}^*$,  which requires
\begin{align*}1&=\Tr{2n}\left(\frac{1}{T_{12}}\right)=\Tr{n}\left(\frac{1}{M_{12}(1+be_1^{-1})}+\frac{1}{M_{12}(1+b^{-1}e_1^{-1})}\right)\\
&=\begin{cases}\Tr{n}\left(\frac{1}{M_{12}}\right), \text{ if } e_1=1,\\
\Tr{n}\left(\frac{E_1}{M_{12}(B+E_1)}\right), \text{ if } e_1\neq
1,\end{cases}
\end{align*} where $E_1=\frac{1}{e_1+e_1^{-1}}.$ Since
$\Tr{n}\left(\frac{1}{M_{12}}\right)=0$,  the first condition requires that $$\Tr{n}\left(\frac{1}{M_{12}}\cdot\frac{B}{(B+E_1)}\right)=1.$$
Letting $e_2=\frac{M_2}{M_1},$ the third condition of \eqref{0003}
gives
$$t_1=\sqrt{\frac{\alpha^{-1}+e_2}{\alpha+e_2}}.$$
 For
every $e_1\in \GF{q}^\star$, the two conditions
$\Tr{n}\left(\frac{1}{M_{12}}\cdot\frac{B}{(B+E_1)}\right)=1$  and
$\Tr{n}\left(\frac{1}{M_{12}}\right)=0$ give
$\#S_{q,\frac{B}{(B+E_1)},0,1}=\frac{q}{4}$ values for $M_{12}$
which leads to $2\times \frac{q}{4}=\frac{q}{2}$ values for
$\frac{z_2}{z_1}$.

 If
$z_2=z_1^{-1}$, then $M_{11}=0$, $e_2=1$, $t_2=t_1^{-1}$ and
$T_{11}=0$. The second condition of \eqref{0003} is satisfied.
Therefore, the number of solutions of \eqref{0003} such that
$M_{12}\neq0$ and $z_2=z_1^{-1}$ is
$$\mathcal{N}_{22}=\#\GF{q}^{\star}\times \frac{q}{2}\times 2=q(q-2).$$

Now, assume $M_1\neq M_2$ and set $\alpha:=\frac{t_2}{t_1}$,
$\beta=\frac{z_2}{z_1}$. Then, $T_{12}=\alpha+\alpha^{-1}$,
$M_{12}=\beta+\beta^{-1}$, $T_{11}=\alpha t_1^2+(\alpha
t_1^2)^{-1}=\frac{1+\alpha
e_2}{\alpha+e_2}+\frac{\alpha+e_2}{1+\alpha e_2}=\frac{(1+\alpha^2)(
1+e_2^2)}{(\alpha+e_2)(1+\alpha e_2)}=\frac{1}{T_{12}^{-1}+E_2}$,
where $E_2=\frac{1}{e_2+e_2^{-1}}.$ Further, it can be checked
 $$z_1=\sqrt{\frac{\beta^{-1}+e_2}{\beta+e_2}}$$ and
$M_{11}=\frac{1}{M_{12}^{-1}+E_2}$. By using
$\frac{M_{11}+T_{11}}{M_{12}+T_{12}}=\frac{\frac{M_{12}}{1+E_2M_{12}}+\frac{T_{12}}{1+E_2T_{12}}}{M_{12}+T_{12}}=\frac{1}{(1+E_2M_{12})(1+E_2T_{12})},$
the second condition of \eqref{0003} can be restated as
$1+(1+E_2M_{12})(1+E_2T_{12})+(1+E_2M_{12})(1+E_2T_{12}^q)=0$, i.e.,
$$M_{12}E_2+(M_{12}E_2)^2=\frac{M_{12}}{T_{12}+T_{12}^q}=Be_1.$$

Set $\gamma_1:=M_{12}E_2$, $x:=\frac{1}{M_{12}}$,
$\gamma_2=\frac{B}{B+E_1}=\frac{B}{B+\frac{1}{\frac{\gamma_1+\gamma_1^2}{B}+\frac{B}{\gamma_1+\gamma_1^2}}}
=\frac{\gamma_1^2+\gamma_1^4+B^2}{\gamma_1+\gamma_1^4+B^2}$. The
following two facts can be easily checked : 1) $\gamma_1=\gamma_2$
if and only if $\gamma_1=\sqrt{B}$; 2) $\gamma_1=\gamma_2+1$ if and
only if $\gamma_1=\sqrt{B}+1$. The condition is $\Tr{n}(x)=0$,
$\Tr{n}(\gamma_1x)=0$ and $\Tr{n}(\gamma_2x)=1$. Thus, there exists
$x$'s satisfying this condition if and only if $\gamma_1\notin
\{0,1,\sqrt{B},\sqrt{B}+1\}$, and when $\gamma_1\in
\GF{q}\setminus\{0,1,\sqrt{B},\sqrt{B}+1\}$, the number of $x$'s
satisfying that condition is $\#S_{Q,\gamma_1,\gamma_2,
0,0,1}=\frac{q}{8}$ thanks to Lemma~\ref{Scard_ext}. Thus, the
number is
$$\mathcal{N}_{23}=(q-4)\times \frac{q}{8}\times 2\times 2 \times 2=q(q-4).$$

The above discussion gives the value of  $\beta_F(1,b)$ when
$b\in\mu_{q+1}^\star$. Indeed,  we have $$ \sum_{c\in
\mu_{q+1}^\star}\tilde{\beta}_F(1,b,c)=\mathcal{N}_1+\mathcal{N}_{21}+\mathcal{N}_{22}+\mathcal{N}_{23}=2q^2-3q$$
and $$\beta_F(1,b)\overset{Prop. \ref{12b}}{=}q^2+\sum_{c\in
\mu_{q+1}^\star}\tilde{\beta}_F(1,b,c)=3q^2-3q.$$\qed
\end{proof}

\subsection{Solving the problem when $b\in \GF{q^2}\setminus\{\mu_{q+1}\cup\GF{q}\}$}

By a similar approach as proceeded at the beginning of the previous subsection, one can derive that, in this case,
Equation~\eqref{aux_eq_qq} is equivalent to
\begin{equation}\label{aux_eq_not_mu_q_F_q}
b^{q-1}=(M_{12}+T_{12})^{q-1} \text{ and }
b^{-(q+1)}=\left(1+\frac{M_{11}+T_{11}}{M_{12}+T_{12}}\right)^{q+1}\text{
and }
\begin{cases} M_1M_2\neq 0, M_1T_2=M_2T_1\text{, or, } \\
z_1=z_2=1, T_1^{q-1}=T_2^{q-1},
\end{cases}
\end{equation}
where $$t_1\neq t_2.$$

\begin{lemma}
If $b\in\GF{q^2}\setminus\{\mu_{q+1}\cup\GF{q}\}$ , then
$\mathcal{N}:=\sum_{c\in \mu_{q+1}^\star}\tilde{\beta}_F(1,b,c)=q^2-3q$ and
$\beta_F(1,b)=2q^2-3q.$
\end{lemma}
\begin{proof}

First, let us consider the case when  $M_1=M_2=0$, i.e. $z_1=z_2=1$. In this
case, we have to compute
\[\#\{(t_1,t_2)\in \mu_{q^2+1}^\star \times \mu_{q^2+1}^\star\mid
b^{q-1}=T_{12}^{q-1}, b^{-(q+1)}=
\left(\frac{T_{11}+T_{12}}{T_{12}}\right)^{q+1},\left(\frac{T_1}{T_2}\right)^{q-1}=1
\},\] that is, the number $\mathcal{N}_1$ of $(t_1,t_2)\in
\mu_{q^2+1}^\star \times \mu_{q^2+1}^\star$ such that
\[
T_{12}=be_1^{-1},T_2=e_2T_1 \text{ for some } e_1, e_2 \in \GF{q}^*
\text{ and } b(T_{11}+T_{12})=fT_{12}  \text{ for some } f \in
\mu_{q+1} .\] The first condition
 $T_{12}=be_1^{-1}$ for some $e_1\in \GF{q}^*$ requires $\Tr{2n}\left(\frac{e_1}{b}\right)=1$, i.e.,
\begin{equation}\label{1fcond}
\Tr{n}\left(e_1\cdot \frac{b+b^{q}}{b^{q+1}}\right)=1.
\end{equation} Under
this condition, as one of two solutions of $T_{12}=be_1^{-1}$, we
can fix
$$\alpha=\frac{t_2}{t_1}\in \mu_{q^2+1}^\star.$$ Then, the second condition $T_2=e_2T_1$ gives $\alpha
t_1+\alpha^{-1}t_1^{-1}=e_2t_1+e_2t_1^{-1}$, i.e.,
\begin{equation}\label{1t1}
t_1=\sqrt{\frac{\alpha^{-1}+e_2}{\alpha+e_2}}.
\end{equation}

If $e_2=1$, then $t_2=\frac{1}{t_1}$, $T_{11}=0$ and so the third
condition gives $b=f\in \mu_{q+1}$ which leads to a contradiction. Hence,
we can assume $e_2\neq 1$.

By using $T_{12}=\alpha+\alpha^{-1}$ and $T_{11}=\alpha
t_1^2+(\alpha t_1^2)^{-1}=\frac{1+\alpha
e_2}{\alpha+e_2}+\frac{\alpha+e_2}{1+\alpha e_2}=\frac{(1+\alpha^2)(
1+e_2^2)}{(\alpha+e_2)(1+\alpha e_2)}=\frac{1}{A+E_2}$, where
$A=\frac{1}{T_{12}}$ and $E_2=\frac{1}{e_2+e_2^{-1}}$, one can
easily verify that the third condition
$\left(\frac{b(T_{11}+T_{12})}{T_{12}}\right)^{q+1}=1$ is equivalent
to $\left(\frac{bE_2}{A+E_2}\right)^{q+1}=1.$ Substituting  $A=b^{-1}e_1$ into the last equality and rearranging using routine computations lead to a quadratic equation

$$U^2+U+\xi(b)=0$$ where
$$U:=\frac{b^{q+1}(b^{q+1}+1)}{e_1(b+b^q)}\cdot E_2$$ and
$$\xi(b):=\frac{b^{q+1}(b^{q+1}+1)}{(b+b^q)^2}.$$

Therefore, if $\Tr{n}\left(\xi(b)\right)=1$, then $\mathcal{N}_1=0$.
Assume $\Tr{n}\left(\xi(b)\right)=0$ and let $u$ be a solution to
the above quadratic equation. Let $x=\frac{b+b^{q}}{b^{q+1}}e_1$,
$\gamma_1= \frac{u}{b^{q+1}+1}$, and $\gamma_2=
\frac{u+1}{b^{q+1}+1}$. Then $E_2=\gamma_1 x$, or, $E_2=\gamma_2 x.$
There exist two elements $e_1$ and $e_2$  of $ GF{q}^*$  satisfying all the three
conditions if and only if it holds $\Tr{n}(x)=1$ and
$\Tr{n}(\gamma_1 x)=0$, or, $\Tr{n}(x)=1$ and $\Tr{n}(\gamma_2
x)=0$. For each of such $x$'s (i.e. $e_1$'s), there are two
corresponding $e_2$'s and two $\alpha$'s, each of which gives a
solution \eqref{1t1}, and in total, 4 solutions of
\eqref{aux_eq_not_mu_q_F_q}. It is easy to check: 1) $\gamma_i=1$
for some $i\in \{1,2\}$ if and only if $b+b^q=1$; 2)
$\gamma_1=\gamma_2$ never hold. Therefore, thanks to
Lemma~\ref{Scard}, we have
\[\mathcal{N}_1=
\begin{cases} 0, \text{ if } \Tr{n}\left(\xi(b)\right)=1,\\
S_{q,\gamma_i,1,0}\times 4=q, \text{ if } \Tr{n}\left(\xi(b)\right)=0 \text{ and } b+b^q=1,\\
(S_{q,\gamma_1,1,0}+S_{q,\gamma_2,1,0}) \times 4=2q, \text{ if }
\Tr{n}\left(\xi(b)\right)=0 \text{ and } b+b^q\neq1.
\end{cases}
\]

This brings us to compute the number $\mathcal{N}_2$ of
$(z_1,t_1),(z_2,t_2)\in \left(\GF{q}^\star\times
\mu_{q^2+1}^\star\right)^2$ such that
\begin{equation}\label{1111}
b^{-2}=(M_{12}+T_{12})^{q-1},
b^{-(q+1)}=\left(1+\frac{M_{11}+T_{11}}{M_{12}+T_{12}}\right)^{q+1},
T_2=\frac{M_2}{M_1}T_1 .
\end{equation}

If $M_{12}=0$, i.e., $z_1=z_2$, then the third condition gives
$t_2=t_1^{-1}$ and $T_{11}=0$. In this case, the first condition can
again be written as $T_{12}=be_1^{-1}$ for some $e_1\in \GF{q}^*$ which
requires ~\eqref{1fcond}. Therefore, the second condition gives
$\left(\frac{b(T_{12}+M_{11})}{T_{12}}\right)^{q+1}=\left(\frac{b(be_1^{-1}+M_{11})}{be_1^{-1}}\right)^{q+1}=(b+e_1M_{11})^{q+1}=1$,
i.e., $$V^2+V+\eta(b)=0,$$ where
$$V:=\frac{b^{q+1}+1}{e_1(b+b^q)M_{11}}$$ and
$$\eta(b):=\frac{b^{q+1}+1}{(b+b^q)^2}=\xi(b^{-1}).$$
Thus, letting $\mathcal{N}_{2}$ be the solution number of
\eqref{1111} with $M_{12}=0,$ if $\Tr{n}\left(\eta(b)\right)=1$,
then $\mathcal{N}_{2}=0$. Assume $\Tr{n}\left(\eta(b)\right)=0$ and
let $v$ be a solution to $V^2+V+\eta(b)=0$. Let
$x=\frac{b+b^{q}}{b^{q+1}}e_1$, $\gamma_1=
\frac{b^{q+1}v}{b^{q+1}+1}$, and $\gamma_2=
\frac{b^{q+1}(v+1)}{b^{q+1}+1}$. Then $M_{11}=\gamma_1 x$, or,
$M_{11}=\gamma_2 x.$ There exist $e_1,z_1\in \GF{q}^*$ satisfying
all the three conditions if and only if it holds $\Tr{n}(x)=1$ and
$\Tr{n}(\gamma_1 x)=0$, or, $\Tr{n}(x)=1$ and $\Tr{n}(\gamma_2
x)=0$. For each of such $x$'s (i.e. $e_1$'s), there are two
corresponding $z_1$'s and two $t_1$'s, and in total, 4 solutions of
\eqref{aux_eq_not_mu_q_F_q}. It is also easy to check : 1)
$\gamma_i=1$ for some $i\in \{1,2\}$ if and only if
$b^{-1}+b^{-q}=1$; 2) $\gamma_1=\gamma_2$ never hold. Therefore,
thanks to Lemma~\ref{Scard}, we have
\[\mathcal{N}_{2}=
\begin{cases} 0, \text{ if } \Tr{n}\left(\eta(b)\right)=1,\\
S_{q,\gamma_i,1,0}\times 4=q, \text{ if } \Tr{n}\left(\eta(b)\right)=0 \text{ and } b^{-1}+b^{-q}=1,\\
(S_{q,\gamma_1,1,0}+S_{q,\gamma_2,1,0}) \times 4=2q, \text{ if }
\Tr{n}\left(\eta(b)\right)=0 \text{ and } b^{-1}+b^{-q}\neq1.
\end{cases}
\]

Now, assume $M_{12}\neq 0$, i.e. $z_1\neq z_2$. The first condition
can be written as $1+\frac{T_{12}}{M_{12}}=be_1^{-1}$ for some $e_1\in
\GF{q}^*$. Since
$\Tr{2n}\left(\frac{1}{T_{12}}\right)=\Tr{n}\left(\frac{1}{M_{12}(1+be_1^{-1})}+\frac{1}{M_{12}(1+b^{q}e_1^{-1})}\right)
=\Tr{n}\left(\frac{1}{M_{12}}\cdot
\frac{1}{1+\frac{e_1^2+b^{q+1}}{e_1(b+b^q)}}\right)$, there exists
$\alpha=\frac{t_2}{t_1}\in \mu_{q^2+1}^\star$ such that
$T_{12}=\alpha+\alpha^{-1}$ if and only if
\begin{equation}\label{2fcond}
\Tr{n}\left(\frac{1}{M_{12}}\cdot
\frac{1}{1+\frac{e_1^2+b^{q+1}}{e_1(b+b^q)}}\right)=1.
\end{equation}
Note $\frac{e_1^2+b^{q+1}}{e_1(b+b^q)}=1$ is equivalent to
$(e_1+b)^{q+1}=0$ which does not hold as $e_1\in \GF{q}$. Letting
$e_2=\frac{M_2}{M_1}$ and $\beta=\frac{z_2}{z_1}$, we have
\begin{equation}\label{2z1}
z_1=\sqrt{\frac{\beta^{-1}+e_2}{\beta+e_2}}
\end{equation}
and the third condition gives
\begin{equation}\label{2t1}
t_1=\sqrt{\frac{\alpha^{-1}+e_2}{\alpha+e_2}}.
\end{equation}

 If
$z_2=z_1^{-1}$, then $M_{11}=0$, $e_2=1$, $t_2=t_1^{-1}$ and
$T_{11}=0$. The second condition can never hold. Therefore,  we can
now assume $M_1\neq M_2$, i.e., $e_2\neq 1$. We can derive
$T_{11}=\alpha t_1^2+(\alpha t_1^2)^{-1}=\frac{1+\alpha
e_2}{\alpha+e_2}+\frac{\alpha+e_2}{1+\alpha e_2}=\frac{(1+\alpha^2)(
1+e_2^2)}{(\alpha+e_2)(1+\alpha e_2)}=\frac{1}{T_{12}^{-1}+E_2}$,
where $E_2=\frac{1}{e_2+e_2^{-1}},$
 and similarly $M_{11}=\frac{1}{M_{12}^{-1}+E_2}$.

Since
$\frac{M_{11}+T_{11}}{M_{12}+T_{12}}=\frac{\frac{M_{12}}{1+E_2M_{12}}+\frac{T_{12}}{1+E_2T_{12}}}{M_{12}+T_{12}}=\frac{1}{(1+E_2M_{12})(1+E_2T_{12})},$
the second condition can be rewritten as
$\left(\frac{1}{(1+E_2M_{12})(1+E_2T_{12})}\right)^{q+1}+\left(\frac{1}{(1+E_2M_{12})(1+E_2T_{12})}\right)^{q}+\frac{1}{(1+E_2M_{12})(1+E_2T_{12})}
=\frac{b^{q+1}+1}{b^{q+1}}$, i.e.,
\begin{align*}1+(1+E_2M_{12})(1+E_2T_{12})+&(1+E_2M_{12})(1+E_2T_{12}^q)=\\&\frac{b^{q+1}+1}{b^{q+1}}\cdot
(1+E_2M_{12})^2(1+E_2T_{12})(1+E_2T_{12}^q).
\end{align*}
Setting $$x:=\frac{1}{M_{12}}, \gamma_1:=M_{12}E_2,
\gamma_2:=\frac{1}{1+\frac{e_1^2+b^{q+1}}{e_1(b+b^q)}}.$$ and
substituting $T_{12}=M_{12}(1+be_1^{-1})$ to above equalities, by
easy computations we obtain
\[1+e_1^{-1}\gamma_1(1+\gamma_1)(b+b^q)=\frac{b^{q+1}+1}{b^{q+1}}\cdot(1+\gamma_1)^2(1+\gamma_1(1+be_1^{-1}))(1+\gamma_1(1+b^qe_1^{-1})),\]
i.e.,
\begin{equation}\label{exxx}
e_1^2[1+(b^{q+1}+1)\gamma_1^4]+e_1[1+(b^{q+1}+1)\gamma_1^2](\gamma_1+\gamma_1^2)(b+b^q)+b^{q+1}(b^{q+1}+1)(\gamma_1^2+\gamma_1^4)=0.
\end{equation}
Here, to ensure the existence of  $\beta, e_1, e_2\in \GF{q}^*$ which
leads to a solution satisfying the three conditions, it has to hold
\begin{equation}\label{2xxx}
\Tr{n}(x)=0, \Tr{n}(\gamma_1x)=0 \text{ and } \Tr{n}(\gamma_2x)=1,
\end{equation}
together with ~\eqref{exxx}.

When $\gamma_1\neq0$ (which follows from $M_{12}\neq 0$ and $E_2\neq
0$, thanks to Lemma~\ref{Scard_ext}), one can see that
Equation~\eqref{2xxx} has exactly $\frac{q}{8}$ solutions if and
only if
\begin{equation}\label{2Lcond}
\gamma_1, \gamma_2,
\gamma_1+\gamma_2\notin \{ 0,1\}
\end{equation}

If $1+(b^{q+1}+1)\gamma_1^2=0$, then $1+(b^{q+1}+1)\gamma_1^4\neq0$,
$e_1^2=\frac{b^{q+1}(b^{q+1}+1)(\gamma_1^2+\gamma_1^4)}{1+(b^{q+1}+1)\gamma_1^4}=b^{q+1}$,
and $\gamma_2=1$. Therefore, we can assume
$1+(b^{q+1}+1)\gamma_1^2\neq0$ and Equation~\eqref{exxx} can be
restated as
\begin{equation}\label{L_equation}
L^2+L=\frac{[1+(b^{q+1}+1)\gamma_1^4]}{[1+(b^{q+1}+1)\gamma_1^2]^2}\cdot
\xi(b),
\end{equation}
where
\begin{equation}\label{e1_equation}
L=e_1\cdot
\frac{1+(b^{q+1}+1)\gamma_1^4}{[1+(b^{q+1}+1)\gamma_1^2](\gamma_1+\gamma_1^2)(b+b^q)}.
\end{equation}
Note
\begin{equation}\label{condd0}
L\neq 1,
\end{equation}
since otherwise \eqref{L_equation} gives $1+(b^{q+1}+1)\gamma_1^4=0$
which contradicts to \eqref{e1_equation}.

Now, set $P:=(b^{q+1}+(L+L^2)(b+b^q)^2)^{\frac{1}{4}}$ and
$Q:=b^{\frac{q+1}{2}}$. Note that $P\neq 0$ for any $L\in \GF{q}$
because
$\frac{P^4}{(b+b^q)^2}=\left(L+\frac{b^q}{b+b^q}\right)\left(L+\frac{b}{b+b^q}\right).$
Then, from~\eqref{L_equation} it follows
\begin{equation}\label{gamma_1eq}
\gamma_1=\frac{P+Q}{P(Q+1)}
\end{equation} and
$\gamma_1(1+\gamma_1)=\frac{(P+Q)Q(P+1)}{P^2(Q+1)^2}$. When
$1+(b^{q+1}+1)\gamma_1^4\neq 0,$ i.e., $L\neq 0$,
from~\eqref{e1_equation} it follows
\begin{align*}e_1&=\frac{L[1+(b^{q+1}+1)\gamma_1^2](\gamma_1+\gamma_1^2)(b+b^q)}{1+(b^{q+1}+1)\gamma_1^4}\\
&=\frac{L[1+(Q^2+1)\cdot\frac{P^2+Q^2}{P^2(Q^2+1)}]\cdot\frac{Q(P+1)(P+Q)}{P^2(Q+1)^2}\cdot(b+b^q)}{1+(Q^2+1)\cdot\frac{P^4+Q^4}{P^4(Q^4+1)}}\\
&=\frac{L(1+\frac{P^2+Q^2}{P^2})\cdot\frac{Q(P+1)(P+Q)}{P^2(Q+1)^2}\cdot(b+b^q)}{1+\frac{P^4+Q^4}{P^4(Q^2+1)}}\\
&=\frac{LQ(P+1)(P+Q)(b+b^q)}{P^4+Q^2}
=\frac{LQ(P+1)(P+Q)(b+b^q)}{(L+L^2)(b+b^q)^2}\\
&=\frac{Q(P+1)(P+Q)}{(L+1)(b+b^q)},
\end{align*}
i.e.,
\begin{equation}\label{e1eq}
e_1= \frac{Q(P+1)(P+Q)}{(L+1)(b+b^q)}.
\end{equation}
When $1+(b^{q+1}+1)\gamma_1^4=0,$ i.e., $L=0,$ one can check that
$e_1$ given by \eqref{e1eq} satisfies ~\eqref{exxx}. Thus, every
pair $(e_1,\gamma_1)$ satisfying ~\eqref{exxx} is completely given
by the variable $L$ via \eqref{gamma_1eq} and \eqref{e1eq}.

Now, we h	ave to determine the values of $L$ satisfying
\eqref{2Lcond}.

\begin{itemize}
\item First, let us show that $\gamma_2\notin \{0,1\}$ is satisfied.
Let $R=(L+1)(b+b^q)^2$. Then, $RL=P^4+Q^2$. From the definition of
$\gamma_2$ and \eqref{e1eq},
\begin{align*}\gamma_2&=\frac{1}{1+\frac{\frac{Q^2(P^2+1)(P^2+Q^2)}{(L^2+1)(b+b^q)^2}+b^{q+1}}{\frac{Q(P+1)(P+Q)}{(L+1)(b+b^q)}\cdot(b+b^q)}}
=\frac{1}{1+\frac{\frac{Q^2(P^2+1)(P^2+Q^2)}{(L+1)(b+b^q)^2}+(L+1)Q^2}{Q(P+1)(P+Q)}}\\
&=\frac{1}{1+\frac{\frac{Q(P^2+1)(P^2+Q^2)}{R}+(L+1)Q}{(P+1)(P+Q)}}
=\frac{R(P+1)(P+Q)}{RP(P+Q+1)+QP^2(Q^2+1)},
\end{align*} i.e.,
\begin{equation}\label{gamma_2eq}
\gamma_2=\frac{R(P+1)(P+Q)}{RP(P+Q+1)+QP^2(Q^2+1)}.
\end{equation}
Then,
\begin{align*}&\gamma_2=1\Longleftrightarrow
R(P+1)(P+Q)=RP(P+Q+1)+QP^2(Q^2+1)\Longleftrightarrow \\
&R=P^2(Q^2+1)\Longleftrightarrow
(L^2+1)(b+b^q)^4=[b^{q+1}+(L+L^2)(b+b^q)^2](b^{q+1}+1)^2,
\end{align*}
which with $U:=b+b^q$ and $V:=b^{q+1}+1$, transforms into
$$L^2U^2(U^2+V^2)+LU^2V^2+U^4+V^2(V+1)=0$$ and so
$$L_1^2+L_1+\frac{(U^2+V^2)[U^4+V^2(V+1)]}{U^2V^4}=0$$ where
$L_1=\frac{LU^2(U^2+V^2)}{U^2V^2}$. However, since
$\frac{(U^2+V^2)[U^4+V^2(V+1)]}{U^2V^4}=\frac{U^2}{V^2}+\frac{U^4}{V^4}+\frac{1}{V}+\frac{1}{V^2}+\frac{V+1}{U^2}
=\left(\frac{U^2}{V^2}+\frac{U^4}{V^4}\right)+\left(\frac{1}{V}+\frac{1}{V^2}\right)+\left(\frac{b}{U}+\frac{b^2}{U^2}\right)$
and $\frac{b}{U}\notin \GF{q}$, it holds
$\Tr{n}\left(\frac{(U^2+V^2)[U^4+V^2(V+1)]}{U^2V^4}\right)\neq 0$
and there does no exist such $L$ in $\GF{q}$.

It is easy to check the following facts.

\item Secondly, 
$\gamma_1=0$ if and only if
\begin{equation}\label{condd1}
 L+L^2=\xi(b).
\end{equation}

\item One has  $\gamma_1=1$ if and only if
\begin{equation}\label{condd4}
   L+L^2=\eta(b).
\end{equation}

\item  $\gamma_1= \gamma_2$ if and only if
\begin{equation}\label{condd2}
L= l_{00}:=\frac{(b+b^q)^2+b^{q+1}+1}{(b+b^q)^2}=1+\eta(b),
\end{equation} since
\begin{align*}&\gamma_1= \gamma_2\Longleftrightarrow R(P+1)P(Q+1)=RP(P+Q+1)+QP^2(Q^2+1)\\
&\Longleftrightarrow R=Q^2+1\Longleftrightarrow
(L+1)(b+b^q)^2=(b^{q+1}+1)\\
&\Longleftrightarrow L=\frac{(b+b^q)^2+b^{q+1}+1}{(b+b^q)^2};
\end{align*}

\item  $\gamma_1=\gamma_2+1$ if and only if
\begin{equation}\label{condd3}
L= l_{01}:=\frac{(b+b^q)^2+b^{q+1}(b^{q+1}+1)}{(b+b^q)^2}=1+\xi(b),
\end{equation}
since $\gamma_1+1=\frac{Q(P+1)}{P(Q+1)}$ and
\begin{align*}&\gamma_2=\gamma_1+1\Longleftrightarrow R(P+Q)(Q+1)=[R(P+Q+1)+QP(Q^2+1)]Q\\
&\Longleftrightarrow
R=Q^2(Q^2+1)\Longleftrightarrow(L+1)(b+b^q)^2=(b^{q+1}+1)b^{q+1}\\
&\Longleftrightarrow
L=\frac{(b+b^q)^2+b^{q+1}(b^{q+1}+1)}{(b+b^q)^2}.\end{align*}
\end{itemize}

Now, consider the following facts :

\begin{enumerate}
\item $l_{00}$ is a solution of ~\eqref{condd1} if and only if
$b+b^q=1$. Indeed,
\begin{align*}
&\eta(b)+\eta(b)^2=\xi(b)\Longleftrightarrow
\frac{b^{q+1}+1}{(b+b^q)^2}+\frac{(b^{q+1}+1)^2}{(b+b^q)^4}=\frac{b^{2(q+1)}+b^{q+1}}{(b+b^q)^2}\\
&\Longleftrightarrow
(b^{q+1}+1)(b+b^q)^2+(b^{q+1}+1)^2=b^{2(q+1)}+b^{q+1}\Longleftrightarrow
b+b^q=1.
\end{align*}
\item $l_{01}$ is a solution of ~\eqref{condd4} if and only if
$b^{-1}+b^{-q}=1$. This follows from the above fact and
$\xi(b)=\eta(b^{-1}).$

\item $l_{00}$ can't be a solution of ~\eqref{condd4} and $l_{01}$
can't be a solution of ~\eqref{condd1}.
\end{enumerate}

Let $l_{10}$ and $l_{11}$ be two solutions of
Equation~\eqref{condd1} and assume $l_{10}\neq l_{00}$. And, let
$l_{20}$ and $l_{21}$ be two solutions of Equation~\eqref{condd4}
and assume $l_{20}\neq l_{01}$. Then, the set $\mathcal{L}$ of
values $L$ satisfying \eqref{2Lcond} is as follows.

\begin{itemize}
\item If $\Tr{n}\left(\xi(b)\right)=1$ and
$\Tr{n}\left(\eta(b)\right)=1$, then $\mathcal{L}= \GF{q}\setminus
\{1,l_{00},l_{01}\}$ and $\#\mathcal{L}=q-3$.
\item If $\Tr{n}\left(\xi(b)\right)=0$, $b+b^q=1$, and
$\Tr{n}\left(\eta(b)\right)=1$, then $\mathcal{L}= \GF{q}\setminus
\{1,l_{00},l_{01}, l_{10}\}$ and $\#\mathcal{L}=q-4$.
\item If $\Tr{n}\left(\xi(b)\right)=0$, $b+b^q\neq1$, and
$\Tr{n}\left(\eta(b)\right)=1$, then $\mathcal{L}= \GF{q}\setminus
\{1,l_{00},l_{01}, l_{10}, l_{11}\}$ and $\#\mathcal{L}=q-5$.
\item If $\Tr{n}\left(\xi(b)\right)=1$,
$\Tr{n}\left(\eta(b)\right)=0$, and $b^{-1}+b^{-q}=1$, then
$\mathcal{L}= \GF{q}\setminus \{1,l_{00},l_{01}, l_{20}\}$ and
$\#\mathcal{L}=q-4$.
\item If $\Tr{n}\left(\xi(b)\right)=1$,
$\Tr{n}\left(\eta(b)\right)=0$, and $b^{-1}+b^{-q}\neq 1$, then
$\mathcal{L}= \GF{q}\setminus \{1,l_{00},l_{01}, l_{20}, l_{21}\}$
and $\#\mathcal{L}=q-5$.
\item If $\Tr{n}\left(\xi(b)\right)=0$,
$\Tr{n}\left(\eta(b)\right)=0$, $b+b^q=1$, and $b^{-1}+b^{-q}\neq
1$, then $\mathcal{L}= \GF{q}\setminus \{1,l_{00},l_{01}, l_{10},
 l_{20}, l_{21}\}$ and $\#\mathcal{L}=q-6$.
\item If $\Tr{n}\left(\xi(b)\right)=0$,
$\Tr{n}\left(\eta(b)\right)=0$, $b+b^q\neq 1$, and $b^{-1}+b^{-q}=
1$, then $\mathcal{L}= \GF{q}\setminus \{1,l_{00},l_{01}, l_{10},
l_{11}, l_{20}\}$ and $\#\mathcal{L}=q-6$.
\item If $\Tr{n}\left(\xi(b)\right)=0$,
$\Tr{n}\left(\eta(b)\right)=0$, $b+b^q\neq 1$, and
$b^{-1}+b^{-q}\neq 1$, then $\mathcal{L}= \GF{q}\setminus
\{1,l_{00},l_{01}, l_{10}, l_{11}, l_{20}, l_{21}\}$ and
$\#\mathcal{L}=q-7$.
\end{itemize}

For each $x$ satisfying \eqref{2xxx}, two $e_2$, two $e_1$ (two
$\alpha$'s) and two $\beta$'s are corresponded respectively which
with~\eqref{2z1} and \eqref{2t1} lead to 8 solutions of
\eqref{aux_eq_not_mu_q_F_q}. Finally, for any case, we have
$$\mathcal{N}=N_1+N_2+\#\mathcal{L}\times S_{q,\gamma_1,\gamma_2, 0,0,1}\times
8=q^2-3q.$$ \qed
\end{proof}

\subsection{Solving the problem when $b\in \GF{q^4}\setminus \GF{q^2}$}
The following determines the set of $b\in \GF{q^4}$ such that
$\sum_{c\in \mu_{q+1}^\star}\tilde{\beta}_F(1,b,c)=0$.
\begin{proposition}\label{tr0}
If $b\in \GF{q^4}\setminus \GF{q^2}$ satisfy \eqref{aux_eq0}, then
$\mathbf{Tr}_n^{4n}(b)=0$. Thus, if $b\in \{\mathfrak{S}_2 \cup
\mathfrak{S}_2'\}$, then $\sum_{c\in
\mu_{q+1}^\star}\tilde{\beta}_F(1,b,c)=0$ and $\beta_F(1,b)=\begin{cases}2,
\text{ if } b\in \mathfrak{S}_2\\
0, \text{ if } b\in \mathfrak{S}_2'.\end{cases}$
\end{proposition}
\begin{proof} Since
$\frac{1}{1+z_1t_1}+\frac{1}{1+z_2t_2}+\left(\frac{1}{1+z_1t_1}+\frac{1}{1+z_2t_2}\right)^{q^2}=
\frac{1}{1+z_1t_1}+\frac{1}{1+z_1t_1^{-1}}+\frac{1}{1+z_2t_2}+\frac{1}{1+z_2t_2^{-1}}=
\frac{z_1t_1+z_1t_1^{-1}}{1+z_1^2+z_1t_1+z_1t_1^{-1}}+\frac{z_2t_2+z_2t_2^{-1}}{1+z_2^2+z_2t_2+z_2t_2^{-1}}=\frac{T_1}{M_1+T_1}+\frac{T_2}{M_2+T_2}=\frac{M_1}{M_1+T_1}+\frac{M_2}{M_2+T_2},$
by using (the second equation of) \eqref{aux_eq0} and \eqref{ztd} we
have
$b+b^{q^2}=(M_1+T_1)^{1-q}\cdot\left(\frac{M_1}{M_1+T_1}\right)^2+(M_2+T_2)^{1-q}\cdot
\left(\frac{M_2}{M_2+T_2}\right)^2=\frac{M_1^2}{(M_1+T_1)^{q+1}}+\frac{M_2^2}{(M_2+T_2)^{q+1}}\in
\GF{q}$, and hence
$\mathbf{Tr}_n^{4n}(b)=b+b^{q^2}+(b+b^{q^2})^q=0.$\qed
\end{proof}
It remained to prove the case  of $b\in\GF{q^4}\setminus
\{\mathfrak{S}_2 \cup \mathfrak{S}_2'\cup \GF{q^2}\}$.
Preliminarily, we can state the following fact.
\begin{proposition}\label{f0}
$$\sum_{b\in\GF{q^4}\setminus \{\mathfrak{S}_2 \cup \mathfrak{S}_2'\cup
\GF{q^2}\}}\sum_{c\in \mu_{q+1}^\star}\tilde{\beta}_F(1,b,c)=(q^3-q^2)\cdot
q\cdot (q-2).$$
\end{proposition}
\begin{proof}
By Theorem~\ref{BaseThm}, one has $\sum_{b\in \GF{q^4}^*}\sum_{c\in
\mu_{q+1}^\star}\tilde{\beta}_F(1,b,c)=q(q^2-q)(q^2-q-1)$ and therefore
\begin{align*}&\sum_{b\in \GF{q^4}\setminus \GF{q^2}}\sum_{c\in
\mu_{q+1}^\star}\tilde{\beta}_F(1,b,c)=\sum_{b\in \GF{q^4}^*}\sum_{c\in
\mu_{q+1}^\star}\tilde{\beta}_F(1,b,c)-\sum_{b\in \GF{q^2}^*}\sum_{c\in
\mu_{q+1}^\star}\tilde{\beta}_F(1,b,c)\\
&=\sum_{c\in \mu_{q+1}^\star}\sum_{b\in
\GF{q^4}^*}\tilde{\beta}_F(1,b,c)-\sum_{b\in \GF{q^2}^*}\sum_{c\in
\mu_{q+1}^\star}\tilde{\beta}_F(1,b,c)\\
&=q(q^2-q)(q^2-q-1)-q^2(q-1)^2=(q^3-q^2)\cdot q\cdot
(q-2).
\end{align*}\qed
\end{proof}

Finally, we shall prove :
\begin{lemma}
If $b\in\GF{q^4}\setminus \{\mathfrak{S}_2 \cup \mathfrak{S}_2'\cup
\GF{q^2}\}$, that is, $b\in \GF{q^4}\setminus\GF{q^2}$ satisfy
$\mathbf{Tr}_n^{4n}(b)=0$, then $\sum_{c\in
\mu_{q+1}^\star}\tilde{\beta}_F(1,b,c)=q^2-2q$ and $\beta_F(1,b)=q^2-2q.$
\end{lemma}
\begin{proof} By exploiting~\eqref{ztd}, the equation~\eqref{aux_eq0} can be rewritten as
\begin{equation}\label{aux_eqq}
\begin{cases}b=(M_1+T_1)^{1-q}\left(\frac{1}{1+z_1t_1}+\frac{1}{1+z_2t_2}\right)^{2q^2}\\
(M_1+T_1)^{q-1}=(M_2+T_2)^{q-1},
\end{cases}
\end{equation}
For any fixed $e\in \mu_{q+1}^\star$, set
\begin{equation}\label{u}
u:=(be)^{\frac{1}{2q^2}}
\end{equation}
and we will consider the solutions $((z_1,t_1),(z_2,t_2))\in
\left(\mu_{q-1} \times\mu_{q^2+1}^\star\right)^2$ of equation
\begin{equation}\label{c}
u=\frac{1}{1+z_1t_1}+\frac{1}{1+z_2t_2}.
\end{equation}

First, let us show that if $((z_1,t_1),(z_2,t_2))\in \left(\mu_{q-1}
\times\mu_{q^2+1}^\star\right)^2$ is a solution to
Equation~\eqref{c}, then it holds
\begin{equation}\label{e}
e=(M_1+T_1)^{q-1}=(M_2+T_2)^{q-1}.
\end{equation}
In fact, if we set $e_1:=(M_1+T_1)^{q-1}/e$, then $e_1\in \mu_{q+1}$
and it holds $b=\frac{1}{e}u^{2q^2}=e_1b',$ where
$b'=(M_1+T_1)^{1-q}\left(\frac{1}{1+z_1t_1}+\frac{1}{1+z_2t_2}\right)^{2q^2}$.
It was shown in the proof of Proposition~\ref{tr0} that
$\mathbf{Tr}_n^{4n}(b')=0$. Since
$\mathbf{Tr}_n^{4n}(b)=e_1\mathbf{Tr}_n^{4n}(b')+\left(e_1+\frac{1}{e_1}\right)(b'+b'^{q^2})^{q}=\left(e_1+\frac{1}{e_1}\right)(b'+b'^{q^2})^{q}=0$
and $b'+b'^{q^2}=e_1(b+b^{q^2})\neq0$, it holds $e_1+\frac{1}{e_1}$
i.e., $e_1=1.$ Also note
\begin{equation}\label{c-cond}
\mathbf{Tr}_n^{4n}(u)=e\mathbf{Tr}_n^{4n}(b)+\left(e+\frac{1}{e}\right)(b+b^{q^2})^{q}=\left(e+\frac{1}{e}\right)(b+b^{q^2})^{q}\neq0.
\end{equation}

Now, if $((z_1,t_1),(z_2,t_2))\in \left(\mu_{q-1}
\times\mu_{q^2+1}^\star\right)^2$ is a solution to
Equation~\eqref{c}, then
$z_2t_2=\frac{(u+1)(1+z_1t_1)+1}{u(1+z_1t_1)+1}=\frac{(u+1)z_1t_1+u}{uz_1t_1+u+1}$
and therefore we have
\begin{equation}\label{z2t2}
z_2=\left(\frac{S_2}{S_1}\right)^{\frac{q^2+1}{2}},
t_2=\left(\frac{S_1}{S_2}\right)^{\frac{q^2-1}{2}},
\end{equation}
where $v=u+1, S_1=uz_1t_1+v, S_2=vz_1t_1+u.$ Since
$M_2+T_2=\left(\frac{S_2}{S_1}\right)^{\frac{q^2+1}{2}}+\left(\frac{S_1}{S_2}\right)^{\frac{q^2+1}{2}}+\left(\frac{S_1}{S_2}\right)^{\frac{q^2-1}{2}}+\left(\frac{S_2}{S_1}\right)^{\frac{q^2-1}{2}}
=\left(\frac{S_1^2+S_2^2}{S_1S_2}\right)^{\frac{q^2+1}{2}}=\frac{(1+z_1t_1)^{q^2+1}}{(S_1S_2)^{\frac{q^2+1}{2}}}=(M_1+T_1)\cdot
\frac{z_1}{(S_1S_2)^{\frac{q^2+1}{2}}}$, the equality~\eqref{e}
gives $(S_1S_2)^l=1$, and the condition $z_2\in \GF{q}$ gives
$\left(\frac{S_2}{S_1}\right)^l=1.$ Therefore, we get
$S_1^l=S_2^l=1,$ i.e.,
\begin{equation}\label{Beq}
\begin{cases}
(uz_1t_1+v)^{q(q^2+1)}=(uz_1t_1+v)^{q^2+1}\\
(vz_1t_1+u)^{q(q^2+1)}=(vz_1t_1+u)^{q^2+1}.
\end{cases}
\end{equation}

Below, we shall prove that Equation~\eqref{Beq} has exactly $q$
solutions $(z_1,t_1)\in \mu_{q-1}\times\mu_{q^2+1}^\star$ for some
$q-2$ values of $u$ and has no solution for other 2 values of $u$.
Since
\begin{align*}
&(uz_1t_1+v)^{q^2+1}=(u^{q^2}z_1t_1^{-1}+v^{q^2})(uz_1t_1+v)\\
&=z_1(u^{q^2+1}z_1+v^{q^2+1}z_1^{-1}+uv^{q^2}t_1+u^{q^2}vt_1^{-1}),
\end{align*}
Equation~\eqref{Beq} can be rewritten as
\begin{equation*}
\begin{cases}
(u^{q^2+1}+u^{q^3+q})z_1+(v^{q^2+1}+v^{q^3+q})z_1^{-1}+uv^{q^2}t_1+u^{q^2}vt_1^{-1}+u^qv^{q^3}t_1^q+u^{q^3}v^qt_1^{-q}=0\\
(v^{q^2+1}+v^{q^3+q})z_1+(u^{q^2+1}+u^{q^3+q})z_1^{-1}+vu^{q^2}t_1+v^{q^2}ut_1^{-1}+v^qu^{q^3}t_1^q+v^{q^3}u^qt_1^{-q}=0,
\end{cases}
\end{equation*}
or,
\begin{equation}\label{Beqq}
\begin{cases}
fz_1+gz_1^{-1}+\mathbf{Tr}_n^{4n}(h_1t_1)=0\\
gz_1+fz_1^{-1}+\mathbf{Tr}_n^{4n}(h_2t_1)=0,
\end{cases}
\end{equation}
where $f=u^{q^2+1}+u^{q^3+q},
g=v^{q^2+1}+v^{q^3+q}=f+\mathbf{Tr}_n^{4n}(u),
h_1=uv^{q^2}=u^{q^2+1}+u, h_2=vu^{q^2}=u^{q^2+1}+u^{q^2}$. Note
$f,g\in \GF{q}$ and
\begin{align*}
gh_2+fh_1&=(u^{q^2+1}+u^{q^3+q}+\mathbf{Tr}_n^{4n}(u))(u^{q^2+1}+u^{q^2})+(u^{q^2+1}+u^{q^3+q})(u^{q^2+1}+u)\\
&=u^{q^2+1}\mathbf{Tr}_n^{4n}(u)+u^{q^2}\mathbf{Tr}_n^{4n}(u)+u^{2q^2+1}+u^{q^3+q^2+q}+u^{q^2+2}+u^{q^3+q+1}\\
&=u^{q^2+q+1}+u^{q^3+q^2+1}+u^{q^3+q^2+q}+u^{q^3+q+1}+u^{q^2}\mathbf{Tr}_n^{4n}(u)\\
&=\mathbf{Tr}_n^{4n}(u^{q^2+q+1})+u^{q^2}\mathbf{Tr}_n^{4n}(u).
\end{align*}
Hence, eliminating the $z_1^{-1}$ term from \eqref{Beqq} gives
\begin{equation}\label{z1}
z_1=\frac{\mathbf{Tr}_n^{4n}(at_1)}{\mathbf{Tr}_n^{4n}(u)^2},
\end{equation}
where
$$a:=\mathbf{Tr}_n^{4n}(u^{q^2+q+1})+u^{q^2}\mathbf{Tr}_n^{4n}(u).$$

Substituting \eqref{z1} into
$z_1+z_1^{-1}=\frac{\mathbf{Tr}_n^{4n}((u+u^{q^2})t_1)}{\mathbf{Tr}_n^{4n}(u)}$
which is obtained by adding the two equations of \eqref{Beqq} yields

\begin{equation}\label{t1_eq}
\mathbf{Tr}_n^{4n}(at_1)^2+\mathbf{Tr}_n^{4n}(u)^4+\mathbf{Tr}_n^{4n}(ht_1)\mathbf{Tr}_n^{4n}(at_1)\mathbf{Tr}_n^{4n}(u)=0,
\end{equation}
where $h=u+u^{q^2}.$

Since
\begin{align*}
&\mathbf{Tr}_n^{4n}(ht_1)\mathbf{Tr}_n^{4n}(at_1)=(ht_1+h^qt_1^q+h^{q^2}t_1^{-1}+h^{q^3}t_1^{-q})
(at_1+a^qt_1^q+a^{q^2}t_1^{-1}+a^{q^3}t_1^{-q})\\
&=\mathbf{Tr}_n^{4n}(hat_1^2)
+\mathbf{Tr}_n^{4n}((ah^q+a^qh)t_1^{q+1})+\mathbf{Tr}_n^{4n}(ha^{q^2})
\end{align*}
and
$$\mathbf{Tr}_n^{4n}(ha^{q^2})=\mathbf{Tr}_n^{4n}\left((u^{q^2}+u)(\mathbf{Tr}_n^{4n}(u^{q^2+q+1})+u\mathbf{Tr}_n^{4n}(u))\right)
=\mathbf{Tr}_n^{4n}(u)^3,$$ Equation~\eqref{t1_eq} transforms into
\begin{equation}\label{t1_eqq}
\mathbf{Tr}_n^{4n}\left(\left(a+\mathbf{Tr}_n^{4n}(u)h\right)at_1^2+\mathbf{Tr}_n^{4n}(u)(ah^q+a^qh)t_1^{q+1}\right)=0.
\end{equation}
Here, $a+\mathbf{Tr}_n^{4n}(u)h=a^{q^2}$, and
\begin{align*}
&\mathbf{Tr}_n^{4n}(u)(ah^q+a^qh)=\mathbf{Tr}_n^{4n}(u)\cdot((\mathbf{Tr}_n^{4n}(u^{q^2+q+1})+u^{q^2}\mathbf{Tr}_n^{4n}(u))
(u^{q^3}+u^q)\\
&+(\mathbf{Tr}_n^{4n}(u^{q^2+q+1})+u^{q^3}\mathbf{Tr}_n^{4n}(u))(u^{q^2}+u))\\
&=\mathbf{Tr}_n^{4n}(u)^2\left(\mathbf{Tr}_n^{4n}(u^{q^2+q+1})+u^{q^3+1}+u^{q^2+q}\right)
=a^{q^2+q}+a^{q^3+1},\end{align*}  and therefore, we get
\begin{equation}\label{Teq}
\mathbf{Tr}_n^{4n}\left(a^{q^2+1}t_1^2+(a^{q^2+q}+a^{q^3+1})t_1^{q+1}\right)=0.
\end{equation}
Notice that both coefficients $a^{q^2+1}$ and $a^{q^2+q}+a^{q^3+1}$
lye in $\GF{q^2}$. Therefore, letting $T_i=t_1^{i}+t_1^{-i}$ for
$i\geq 1$, Equality~\eqref{Teq} is restated as
\begin{equation}\label{Teqq}
a^{q^2+1}T_1^2+a^{q^3+q}T_1^{2q}+(a^{q^2+q}+a^{q^3+1})T_{q+1}+(a^{q^3+q^2}+a^{q+1})T_{q-1}=0.
\end{equation}
 By using $T_{q+1}+T_{q-1}=T_1^{q+1}$ and
$T_{q-1}=T_1^{q+1}\mathbf{Tr}_1^n\left(\frac{1}{T_1}\right)^2$ which
were proved in \cite{KM22-1} (in fact, $T_i$ is Dickson polynomials
and the second equality was initially shown in Lemma 2.1 of
\cite{CM95}), Equation~\eqref{Teqq} again transforms into
$$a^{q^2+1}T_1^2+a^{q^3+q}T_1^{2q}+(a^{q^2+q}+a^{q^3+1})T_1^{q+1}+\mathbf{Tr}_n^{4n}\left(a^{q+1}\right)T_1^{q+1}\mathbf{Tr}_1^n\left(\frac{1}{T_1}\right)^2=0,$$
or, with $X=\frac{1}{T_1}$,
\begin{equation}\label{Xeq}
\alpha X^{q-1}+\alpha^qX^{1-q}+\beta+\mathbf{Tr}_1^n(X)^2=0,
\end{equation}
where
$\alpha=\frac{a^{q^2+1}}{\mathbf{Tr}_n^{4n}\left(a^{q+1}\right)},
\beta=\frac{a^{q^2+q}+a^{q^3+1}}{\mathbf{Tr}_n^{4n}\left(a^{q+1}\right)}$.

Now, we consider the number of the solutions $x\in\GF{q^2}$ to
Equation~\eqref{Xeq}. Note that if $x\in\GF{q^2}$ is a solution to
Equation~\eqref{Xeq}, then
$\mathbf{Tr}_1^{2n}(x)=\mathbf{Tr}_1^{n}(x)+\mathbf{Tr}_1^{n}(x)^q=
\beta+\beta^q=1$.

 We can set $x:=y\lambda$ where $y\in \GF{q}^\star$
and $\lambda\in \mu_{q+1}$. Note $y=x^{\frac{q+1}{2}}$ and
$\lambda=x^{\frac{1-q}{2}}$. Since
$\mathbf{Tr}_1^n(x)+\mathbf{Tr}_1^n(x)^2=x+x^q=y(\lambda+\lambda^{-1})$,
from $\eqref{Xeq}+\eqref{Xeq}^2$ it follows
$$\alpha\lambda^{-2}+\alpha^q\lambda^2+\beta+\left(\alpha\lambda^{-2}+\alpha^q\lambda^2+\beta\right)^2+y^2(\lambda^2+\lambda^{-2})=0,$$
i.e.,
\begin{equation}\label{yeq}
y^2=\frac{\alpha\lambda^{-2}+\alpha^q\lambda^2+\beta+\left(\alpha\lambda^{-2}+\alpha^q\lambda^2+\beta\right)^2}{\lambda^2+\lambda^{-2}}.
\end{equation}

Then, $\mathbf{Tr}_1^{n}(x)^2=\mathbf{Tr}_1^{n}(y^2\lambda^2)=
\mathbf{Tr}_1^{n}\left(\frac{\alpha\lambda^{2}+\alpha^q\lambda^6+\beta\lambda^{4}+\left(\alpha+\alpha^q\lambda^4+\beta\lambda^{2}\right)^2}{\lambda^{4}+1}\right)$.
Letting $\omega=\lambda^{2}$, Equation~\eqref{Xeq} transforms into
\begin{equation}\label{oeq}
\alpha\omega^{q}+\alpha^q\omega+\beta+\mathbf{Tr}_1^{n}\left(\frac{\alpha^{2q}\omega^4+\alpha^q\omega^3+(\beta+\beta^2)\omega^2+\alpha\omega+\alpha^2}{1+\omega^2}\right)=0,
\end{equation}
an equation in terms of the variable $\omega\in \mu_{q+1}$. Since
\begin{align*}
&\mathbf{Tr}_1^{n}\left(\frac{\alpha^{2q}\omega^4+\alpha^q\omega^3+(\beta+\beta^2)\omega^2+\alpha\omega+\alpha^2}{1+\omega^2}\right)\\
&=\mathbf{Tr}_1^{n}\left(\alpha^{2q}\omega^2+\alpha^q\omega+\alpha^{2q}+\beta+\beta^2+\frac{(\alpha+\alpha^q)\omega+\alpha^2+\alpha^{2q}+\beta+\beta^2}{1+\omega^2}\right)\\
&=\alpha^q\omega+\alpha\omega^q+\beta+\beta^q+\mathbf{Tr}_1^{n}(\alpha^{2q})+\mathbf{Tr}_1^{n}\left(\frac{(\alpha+\alpha^q)\omega+\alpha^2+\alpha^{2q}+\beta+\beta^2}{1+\omega^2}\right)\\
&=\alpha^q\omega+\alpha\omega^q+\beta+\beta^q+\mathbf{Tr}_1^{n}(\alpha^{2q})+\mathbf{Tr}_1^{n}\left(\frac{\alpha+\alpha^q}{1+\omega}+\frac{\alpha^2+\alpha^{2q}}{1+\omega^2}+\frac{\alpha+\alpha^q+\beta+\beta^2}{1+\omega^2}\right)\\
&=\alpha^q\omega+\alpha\omega^q+\beta+\beta^q+\mathbf{Tr}_1^{n}(\alpha^{2q})+\frac{\alpha+\alpha^q}{1+\omega}+\frac{\alpha+\alpha^q}{1+\omega^{-1}}+\mathbf{Tr}_1^{n}\left(\frac{\alpha+\alpha^{q}+\beta+\beta^2}{1+\omega^2}\right)\\
&=\alpha^q\omega+\alpha\omega^q+\beta+\beta^q+\mathbf{Tr}_1^{n}(\alpha^{2q})+\alpha+\alpha^q+\mathbf{Tr}_1^{n}\left(\frac{\alpha+\alpha^{q}+\beta+\beta^2}{1+\omega^2}\right),
\end{align*}
from~\eqref{oeq} it follows
\begin{equation}\label{feq}
\mathbf{Tr}_1^{n}\left(\frac{\gamma}{1+\omega^2}\right)=\delta,
\end{equation}
where $\gamma=\alpha+\alpha^{q}+\beta+\beta^2$ and
$\delta=\beta^q+\mathbf{Tr}_1^{n}(\alpha^{2q})+\alpha+\alpha^q$.

It is easy to check $\delta+\delta^2=\gamma\in \GF{q}$ as
$\beta+\beta^q=1$. Further, for any $\omega\in \mu_{q+1}^\star$,
$\mathbf{Tr}_1^{n}\left(\frac{\gamma}{1+\omega^2}\right)+\mathbf{Tr}_1^{n}\left(\frac{\gamma}{1+\omega^2}\right)^2
=\frac{\gamma}{1+\omega^2}+\left(\frac{\gamma}{1+\omega^2}\right)^q
=\frac{\gamma}{1+\omega^2}+\frac{\gamma}{1+\omega^{-2}}=\gamma=\delta+\delta^2$.
Therefore,  for any $\omega\in \mu_{q+1}^\star$,
$\mathbf{Tr}_1^{n}\left(\frac{\gamma}{1+\omega^2}\right)=\delta$ or
$\mathbf{Tr}_1^{n}\left(\frac{\gamma}{1+\omega^2}\right)=\delta+1$.

On the other hand, for any fixed element $\omega_0\in \mu_{q+1}$,
when $\omega$ runs through $\mu_{q+1}^\star$, the elements
$\frac{1}{1+\omega^2}+\frac{1}{1+\omega_0^2}$ run through $\GF{q}$,
 and therefore, when $\gamma\neq 0$,
$\mathbf{Tr}_1^{n}\left(\frac{\gamma}{1+\omega^2}\right)$ takes each
of two values
$\mathbf{Tr}_1^{n}\left(\frac{\gamma}{1+\omega_0^2}\right)$ and
$\mathbf{Tr}_1^{n}\left(\frac{\gamma}{1+\omega_0^2}\right)+1$
exactly $\frac{q}{2}$ times respectively.

Hence, when $\gamma\neq 0$, Equation~\eqref{feq} has exactly
$\frac{q}{2}$ solutions in $\mu_{q+1}^{\star}$ and
Equation~\eqref{oeq} so does. Consequently, with \eqref{yeq},
Equation~\eqref{Xeq} also has $\frac{q}{2}$ solutions in $\GF{q^2}$,
and subsequently Equation~\eqref{Teq} has $2*\frac{q}{2}=q$
solutions in $\mu_{q^2+1}^\star$. This shows that
Equation~\eqref{Beq}, with \eqref{z1}, has $q$ solutions
$(z_1,t_1)\in \GF{q}^\star\times\mu_{q^2+1}^\star$,  and after all,
with \eqref{z2t2}, Equation~\eqref{c} has exactly $q$ solutions
$((z_1,t_1),(z_2,t_2))\in \left(\mu_{q-1}
\times\mu_{q^2+1}^\star\right)^2$.

The condition $\gamma=0$ is equivalent to
$$(a^{q^2+1}+a^{q^3+q}+a^{q^2+q}+a^{q^3+1})\mathbf{Tr}_n^{4n}\left(a^{q+1}\right)=(a^{q^2+q}+a^{q^3+1})^2.$$
By routine computations, one can check this is equivalent to
$$(u^{q^2+1}+u^{q^3+q}+u^{q^2+q}+u^{q^3+1})\mathbf{Tr}_n^{4n}(u^{q+1})=\left(\mathbf{Tr}_n^{4n}(u^{q^2+q+1})+u^{q^2+q}+u^{q^3+1}\right)^2,$$
or,
$$(u^{q^2+1}+u^{q^3+q})\mathbf{Tr}_n^{4n}(u^{q+1})+(u^{q+1}+u^{q^3+q^2})^{q+1}=\left(\mathbf{Tr}_n^{4n}(u^{q^2+q+1})\right)^2,$$
or,
$$\mathbf{Tr}_n^{4n}\left(u^{q^2+q+1}(u+u^q+u^{q^2}+u^{q^2+q+1})\right)=0.$$
Substituting $u=(be)^{\frac{1}{2q^2}}$ to here gives
$$\mathbf{Tr}_n^{4n}\left(b^{q^2+q+1}[b^q+e^2(b+b^{q^2}+b^{q^2+q+1})]\right)=0,$$
i.e.
\begin{equation}\label{ff}
\mathbf{Tr}_n^{4n}\left(b^{q^2+2q+1}\right)+e^2(L+L^{q^2})+e^{-2}(L+L^{q^2})^q=0,
\end{equation}
where $L=b^{q^2+q+2}+b^{2q^2+q+1}+b^{2q^2+2q+2}.$

By using $\mathbf{Tr}_n^{4n}(b)=0$, one can derive
$L+L^{q^2}=b^{q^2+1}(b+b^{q^2})^{2}(1+b^{q^2+1}).$ If $b^{q^2+1}=1$,
then $ \mathbf{Tr}_1^{2n}\left(\frac{1}{b+b^{-1}}\right)=1$ and so
$b+b^{-1}\notin \GF{q}$. Then,
$\mathbf{Tr}_n^{4n}(b)=b+b^{-1}+(b+b^{-1})^q\neq 0$ which is a
contradiction. Hence, $L+L^{q^2}\neq0$ and therefore there are at
most two $e$'s in $\mu_{q+1}^\star$ such that satisfy \eqref{ff},
i.e., $\gamma=0.$ Thus, we get
$$\sum_{c\in \mu_{q+1}^\star}\tilde{\beta}_F(1,b,c)\geq
 q\cdot (q-2)$$ and thus
$$\sum_{b\in\GF{q^4}\setminus \{\mathfrak{S}_2 \cup \mathfrak{S}_2'\cup
\GF{q^2}\}}\sum_{c\in \mu_{q+1}^\star}\tilde{\beta}_F(1,b,c)\geq
(q^3-q^2)\cdot q\cdot (q-2).$$ Proposition~\ref{f0} shows that this
is indeed an equality. Therefore also $$\sum_{c\in
\mu_{q+1}^\star}\tilde{\beta}_F(1,b,c)=
 q\cdot (q-2)$$ and with Proposition~\ref{12b} the lemma is concluded.\qed

\begin{remark} Above proof shows that when $\gamma=0$ Equation~\eqref{feq} has
no solution and therefore $\delta\neq0.$ If $\gamma=\delta=0$, then
$\beta^{q}+\mathbf{Tr}_1^{n}(\alpha^{2q})\overset{\delta=0}{=}\alpha+\alpha^q\overset{\gamma=0}{=}\beta+\beta^2\overset{\beta+\beta^q=1}{=}\beta^q+\beta^{2q}$
and so $\beta^{2q}+\mathbf{Tr}_1^{n}(\alpha^{2q})=0,$ i.e.,
$\beta=\mathbf{Tr}_1^{n}(\alpha)$. This shows that from $\gamma=0$
it follows $\beta=\mathbf{Tr}_1^{n}(\alpha)+1$. \qed
\end{remark}

\end{proof}

\section{Conclusions}
Using algebraic techniques  and finer manipulations over finite fields, we have completely determined the boomerang spectrum of the power
permutation $F(X)=X^{2^{3n}+2^{2n}+2^{n}-1}$ over  the finite field $\GF{2^{4n}}$ of order $2^{4n}$, showing
the boomerang uniformity of that power permutation is
$3(2^{2n}-2^n)$. Theorem  \ref{maintheorem} stated the main technical result of the paper.
More importantly, for any value in the boomerang
spectrum, the set of $b$'s giving this value was explicitly
determined.   Notably,   we have  consequently extracted optimal functions $F$ 
over the large subset  $\mathfrak{S}_2$ over  $\GF{2^{4n}}$,
which corresponds in the framework of (symmetric) cryptography to S-boxes with the optimal value of its
boomerang uniformity $\beta_F(1,b)$ ($b\in\mathfrak{S}_2$), that is, of maximum resistance against boomerang attacks.

\end{document}